\documentclass[11pt,letterpaper,reqno]{amsart}

\title{Symmetry and Conservation Laws in Semiclassical Wave Packet Dynamics}
\author{Tomoki Ohsawa}
\address{Department of Mathematical Sciences, The University of Texas at Dallas, 800 W Campbell Rd, Richardson, TX 75080-3021}
\email{tomoki@utdallas.edu}

\date{\today}



\usepackage{graphicx,mathrsfs,amssymb,paralist}
\usepackage[margin=1in, marginpar=.5in]{geometry}
\usepackage[numbers,sort&compress]{natbib}
\usepackage[colorlinks=false]{hyperref}

\usepackage{subfigure}
\usepackage{tikz-cd}

\theoremstyle{plain}
\newtheorem{theorem}{Theorem}[section]

\newtheorem{lemma}[theorem]{Lemma}
\newtheorem{proposition}[theorem]{Proposition}

\theoremstyle{definition}

\newtheorem{example}[theorem]{Example}

\theoremstyle{remark}
\newtheorem{remark}[theorem]{Remark}

\def\od#1#2{\dfrac{d#1}{d#2}}
\def\pd#1#2{\dfrac{\partial #1}{\partial #2}}

\def\parentheses#1{{\left(#1\right)}}
\def\brackets#1{{\left[#1\right]}}
\def\braces#1{{\left\{#1\right\}}}

\def\tr{\mathop{\mathrm{tr}}\nolimits}

\def\norm#1{{\left\|#1\right\|}}
\def\abs#1{{\left|#1\right|}}
\def\DS{\displaystyle}
\def\R{\mathbb{R}}
\def\C{\mathbb{C}}

\def\defeq{\mathrel{\mathop:}=}

\def\setdef#1#2{{\left\{ #1 \ |\ #2 \right\}}}
\def\ip#1#2{{\left\langle#1,#2\right\rangle}}
\def\tip#1#2{{\langle#1,#2\rangle}}
\def\exval#1{{\left\langle#1\right\rangle}}
\def\texval#1{{\langle#1\rangle}}

\renewcommand{\Re}{\operatorname{Re}}
\renewcommand{\Im}{\operatorname{Im}}
\def\eps{\varepsilon}
\def\SO{\mathsf{SO}}

\def\Sp{\mathsf{Sp}}
\def\U{\mathsf{U}}
\def\so{\mathfrak{so}}

\newenvironment{tbmatrix}{\left[\begin{smallmatrix}}{\end{smallmatrix}\right]}

\def\d{{\bf d}}
\def\ins#1{{\bf i}_{#1}}

\newcommand\Ad{\operatorname{Ad}}

\begin{document}

\footskip=.6in

\begin{abstract}
  We formulate symmetries in semiclassical Gaussian wave packet dynamics and find the corresponding conserved quantities, particularly the semiclassical angular momentum, via Noether's theorem.
  We consider two slightly different formulations of Gaussian wave packet dynamics; one is based on earlier works of Heller and Hagedorn, and the other based on the symplectic-geometric approach by Lubich and others.
  In either case, we reveal the symplectic and Hamiltonian nature of the dynamics and formulate natural symmetry group actions in the setting to derive the corresponding conserved quantities (momentum maps).
  The semiclassical angular momentum inherits the essential properties of the classical angular momentum as well as naturally corresponds to the quantum picture.
\end{abstract}

\maketitle

\section{Introduction}
\subsection{Gaussian Wave Packet and Semiclassical Dynamics}
Techniques involving the Gaussian wave packets have been used to approximate quantum dynamics in the semiclassical regime for a long time.
The origin of the modern usage seems to go back to the works of \citet{He1975a,He1976b,He1981}, \citet{Ha1980,Ha1998}, and \citet[Section~7]{Li1986}; see also references therein.
The most salient feature of the Gaussian wave packet is that it gives an {\em exact} solution to the Schr\"odinger equation with quadratic potentials, given that the parameters governing the dynamics of the wave packet follow a set of ordinary differential equations (ODEs) that is essentially Hamilton's equations of classical mechanics plus some additional ODEs.
It is also established by \citet{Ha1980, Ha1998} that, even if the potential is not quadratic, under some regularity assumptions, the Gaussian wave packet approximates quantum dynamics to the order of $O(\sqrt{\hbar})$ with the same set of ODEs governing the dynamics of it.
\citet{Ha1980, Ha1998} pushes it even further to show that one can generate an orthonormal basis of $L^{2}(\R^{d})$, $d$ being the dimension of the configuration space, by following a construction analogous to the Hermite functions.
Furthermore, it is shown that one may improve the order of approximation by selecting a finite number of elements from the basis, again using the same set of ODEs; such technique is implemented by \citet{FaGrLu2009} to solve the semiclassical Schr\"odinger equation numerically.
The key feature of this basis, often called the Hagedorn wave packets, is that the basis elements are time-dependent and their time evolution is governed by the ODEs.
Therefore, quantitative and qualitative understanding of the behaviors of the solutions of the ODEs is important in theoretical and numerical studies using the Hagedorn wave packets.

\subsection{Symplectic/Hamiltonian View of Gaussian Wave Packet Dynamics}
The symplectic/Hamiltonian nature of the dynamics of the ODEs governing the semiclassical Gaussian wave packet has been studied earlier by, e.g., \citet{Li1988}, \citet{SiSuMu1988}, \citet{BrLaVa1989}, and \citet{PaSc1994}.
More recently, \citet{Lu2008} (see also \citet{FaLu2006}) gave a more systematic account of the symplectic formulation of the Gaussian wave packet dynamics, and this was studied further by us~\cite{OhLe2013} by making use of techniques of geometric mechanics.

The advantage of the symplectic/Hamiltonian viewpoint is that we make use of the arsenal of tools for Hamiltonian systems; see e.g., \citet{AbMa1978} and \citet{MaRa1999}.
Most notably, symmetry and conservation laws in Hamiltonian systems are linked via Noether's theorem: Practically speaking one may easily find conservation laws by observing symmetries in the system of interest.

The symplectic approach to semiclassical dynamics is also natural in the following sense:
Semiclassical dynamics is roughly speaking quantum dynamics in the regime close to classical dynamics; but then the basic equations of classical and quantum dynamics are both Hamiltonian systems with respect to appropriate symplectic structures (see a brief sketch of the symplectic approach to quantum dynamics in Section~\ref{ssec:SymplecticQM}).
So it is natural to formulate semiclassical dynamics retaining the underlying symplectic structure for quantum dynamics as well as to establish a link with the symplectic structure for classical dynamics.

\subsection{Main Results and Outline}
The main focus of this paper is to exploit the symplectic geometry behind the Gaussian wave packet dynamics and derive conserved quantities for the dynamics when the potential has symmetries.
We are mainly interested in the case with rotational symmetry and the corresponding angular momentum in the semiclassical setting.
Needless to say, identifying conserved quantities helps understand the qualitative and quantitative properties of the dynamics.
Moreover, in this particular setting, we will show how the semiclassical dynamics inherits some features of the corresponding classical counterpart, as well as being compatible with the quantum picture; thereby building a bridge between the classical and quantum formalisms.

In Section~\ref{sec:Symplectic_Semiclassical_Dynamics}, we first give a brief review of the symplectic formulation of the Gaussian wave packet dynamics by \citet{Lu2008} and \citet{OhLe2013}; this formulation is more amenable to symmetry analysis because the set of ODEs as a whole is formulated as a single Hamiltonian system.
Then, in Section~\ref{sec:Symmetry_in_Semiclassical_Dynamics}, we consider the Gaussian wave packet dynamics with rotational symmetry, and find and look into the properties of the corresponding conserved quantities (momentum maps).
It turns out that, thanks to the symplectic formulation, the semiclassical angular momentum inherits the key properties of the classical angular momentum as well as possesses a natural link with the angular momentum operator in quantum mechanics.
We also present a simple numerical experiment to illustrate the result.

In Section~\ref{sec:Symmetry_in_Hagedorn}, we switch our focus to the other, more prevalent formulation of the Gaussian wave packet dynamics of \citet{Ha1980, Ha1998} (see also \citet{Heller-LesHouches}).
The underlying symplectic structure for this formulation is not very prominent: It should be extracted from the so-called first variation equation (i.e., the linearization along the solutions) of classical Hamiltonian system.
The symplectic and Hamiltonian nature of the first variation equation, summarized in Appendix~\ref{sec:GeomOfFirstVarEq}, helps us find conserved quantities for the Hagedorn wave packet dynamics in the presence of symmetry, thereby showing a Noether-type theorem for the Hagedorn wave packet dynamics.
Again the main focus is on the rotational symmetry: We find the corresponding conserved quantity, in addition to the classical angular momentum, of the Hagedorn wave packet dynamics.

\section{Symplectic Semiclassical Dynamics}
\label{sec:Symplectic_Semiclassical_Dynamics}
We first give a brief review of the symplectic formulation of the Gaussian wave packet dynamics following our previous work~\cite{OhLe2013}.
\subsection{Symplectic/Hamiltonian Formulation of Quantum Mechanics}
\label{ssec:SymplecticQM}
Let us first give a brief sketch of the symplectic structure for the symplectic/Hamiltonian formulation of the Schr\"odinger equation alluded above.
We will exploit this structure in the next subsection to come up with the corresponding symplectic structure for semiclassical dynamics.

Let $\mathcal{H}$ be a complex (often infinite-dimensional) Hilbert space equipped with a (right-linear) inner product $\ip{\cdot}{\cdot}$ and its induced norm $\norm{\cdot}$.
It is well-known~(see, e.g., \citet[Section~2.2]{MaRa1999}) that the two-form $\Omega$ on $\mathcal{H}$ defined by
\begin{equation*}
  \Omega(\psi_{1}, \psi_{2}) = 2\hbar \Im\ip{\psi_{1}}{\psi_{2}}
\end{equation*}
is a symplectic form, and hence $\mathcal{H}$ is a symplectic vector space.
Given a Hamiltonian operator $\hat{H}$ on $\mathcal{H}$, we may write the expectation value of the Hamiltonian $\texval{\hat{H}}$ as
\begin{equation*}
 \texval{\hat{H}}(\psi) \defeq \texval{\psi, \hat{H}\psi},
\end{equation*}
which defines a real-valued function\footnote{The domain of the Hamiltonian operator $\hat{H}$ is in general not the whole $\mathcal{H}$ and so $\texval{\hat{H}}$ is not defined on the whole $\mathcal{H}$, but we do not delve into these issues here.} on $\mathcal{H}$.
Then the corresponding Hamiltonian flow
\begin{equation*}
  X_{\exval{\hat{H}}}(\psi) = (\psi, \dot{\psi}) \in T\mathcal{H} \cong \mathcal{H} \times \mathcal{H}
\end{equation*}
on $\mathcal{H}$ defined by
\begin{equation*}
  {\bf i}_{X_{\texval{\hat{H}}}} \Omega = {\bf d}\texval{\hat{H}}
\end{equation*}
gives the Schr\"odinger equation
\begin{equation}
  \label{eq:SchroedingerEq-abstract}
  \dot{\psi} = -\frac{{\rm i}}{\hbar}\hat{H} \psi.
\end{equation}

In this paper, $\mathcal{H} = L^{2}(\R^{d})$ and the Hamiltonian operator is
\begin{equation}
  \label{eq:SchroedingerEq-L2}
  \hat{H} = -\frac{\hbar^{2}}{2m} \Delta + V(x),
\end{equation}
and so \eqref{eq:SchroedingerEq-abstract} takes the familiar form
\begin{equation*}
  {\rm i}\hbar\,\pd{}{t}\psi(x,t) = -\frac{\hbar^{2}}{2m} \pd{^{2}}{x_{j}^{2}} \psi(x,t) + V(x)\,\psi(x,t).
\end{equation*}

\subsection{The Gaussian Wave Packet}
We are interested in the dynamics of the Gaussian wave packet given by
\begin{equation}
  \label{eq:chi}
  \chi(y;x) \defeq \exp\braces{ \frac{{\rm i}}{\hbar}\brackets{ \frac{1}{2}(x - q)^{T}(\mathcal{A} + {\rm i}\mathcal{B})(x - q) + p \cdot (x - q) + (\phi + {\rm i} \delta) } },
\end{equation}
where
\begin{equation}
  \label{eq:y}
  y \defeq (q, p, \mathcal{A}, \mathcal{B}, \phi, \delta)  
\end{equation}
parametrizes the Gaussian wave packet; more specifically the Gaussian wave packet $\chi$ depends on the time through the set of parameters $y$.
These parameters live in the following spaces: $(q,p) \in T^{*}\R^{d}$, $\phi \in \mathbb{S}^{1}$, $\delta \in \R$, and $\mathcal{C} \defeq \mathcal{A} + {\rm i}\mathcal{B}$ is a $d \times d$ complex symmetric matrix with a positive-definite imaginary part, i.e., the matrix $\mathcal{C}$ is an element in the {\em Siegel upper half space}~\cite{Si1943} (see also the brief summary in Appendix~\ref{sec:Sigma_d} of this paper) defined by
\begin{equation}
  \label{eq:Sigma_d}
  \Sigma_{d} \defeq 
  \setdef{ \mathcal{C} = \mathcal{A} + {\rm i}\mathcal{B} \in \mathbb{C}^{d\times d} }{ \mathcal{A}, \mathcal{B} \in \text{Sym}_{d}(\R),\, \mathcal{B} > 0 },
\end{equation}
where $\text{Sym}_{d}(\R)$ denotes the set of $d \times d$ real symmetric matrices, and $\mathcal{B} > 0$ means that $\mathcal{B}$ is positive-definite.
It is easy to see that the (real) dimension of $\Sigma_{d}$ is $d(d+1)$.
Since the dynamics of the above Gaussian wave packet is governed by the set of parameters $y$, we are essentially looking at the dynamics in the $(d+1)(d+2)$-dimensional manifold
\begin{equation*}
 \mathcal{M} \defeq T^{*}\R^{d} \times \Sigma_{d} \times \mathbb{S}^{1} \times \R = \{(q, p, \mathcal{A}, \mathcal{B}, \phi, \delta)\}.
\end{equation*}

The normalized version of $\chi$ is given by
\begin{equation}
  \label{eq:psi_0}
  \psi_{0} \defeq \frac{\chi}{\norm{\chi}} = \parentheses{ \frac{\det\mathcal{B}}{(\pi\hbar)^{d}} }^{1/4} \exp\braces{ \frac{{\rm i}}{\hbar}\brackets{ \frac{1}{2}(x - q)^{T}(\mathcal{A} + {\rm i}\mathcal{B})(x - q) + p \cdot (x - q) + \phi } }.
\end{equation}
Ignoring the phase factor $e^{{\rm i}\phi/\hbar}$ in the above expression corresponds to taking the equivalence class (i.e., quantum state) $[\psi_{0}]_{\mathbb{S}^{1}}$ in the projective Hilbert space $\mathbb{P}(L^{2}(\R^{d})) \defeq \mathbb{S}(L^{2}(\R^{d}))/\mathbb{S}^{1}$, where $\mathbb{S}(L^{2}(\R^{d}))$ is the unit sphere in $L^{2}(\R^{d})$, i.e., the set of normalized wave functions, and the quotient is defined by the $\mathbb{S}^{1}$ phase action $\psi \mapsto e^{{\rm i}\theta} \psi$.
Hence a simple representative element in $L^{2}(\R^{d})$ for $[\psi_{0}]_{\mathbb{S}^{1}} \in \mathbb{P}(L^{2}(\R^{d}))$ would be $\psi_{0}$ without the phase:
\begin{equation}
  \label{eq:psi_0-nophase}
  [\psi_{0}]_{\mathbb{S}^{1}} = \brackets{ \parentheses{ \frac{\det\mathcal{B}}{(\pi\hbar)^{d}} }^{1/4} \exp\braces{ \frac{{\rm i}}{\hbar}\brackets{ \frac{1}{2}(x - q)^{T}(\mathcal{A} + {\rm i}\mathcal{B})(x - q) + p \cdot (x - q) } } }_{\mathbb{S}^{1}}.
\end{equation}

\subsection{Hamiltonian Dynamics of Semiclassical Wave Packet}
So how should the parameters $y$ in \eqref{eq:y} evolve in time so that \eqref{eq:chi} gives the best approximation to the solutions of the Schr\"odinger equation~\eqref{eq:SchroedingerEq-L2}?
The key idea due to \citet[Section~II.1]{Lu2008} is to view the Gaussian wave packet~\eqref{eq:chi} as an embedding of $\mathcal{M}$ to $\mathcal{H} \defeq L^{2}(\R^{d})$ defined by
\begin{equation*}
  \iota\colon \mathcal{M} \hookrightarrow \mathcal{H};
  \qquad
  \iota(y) = \chi(y;\cdot)
\end{equation*}
One can then show that $\mathcal{M}$ is a symplectic manifold with symplectic form $\Omega_{\mathcal{M}} \defeq \iota^{*}\Omega$ pulled back from $\mathcal{H} = L^{2}(\R^{d})$.
Similarly, define a Hamiltonian $H\colon \mathcal{M} \to \R$ by the pull-back $H \defeq \iota^{*} \texval{\hat{H}}$, i.e.,
\begin{equation}
  \label{eq:H}
  H(y) = \tip{\chi(y,\cdot)}{\hat{H}\chi(y,\cdot)}.
\end{equation}
See \citet[Section~3.2]{OhLe2013} for their coordinate expressions.
Then we may define a Hamiltonian system on $\mathcal{M}$ by
\begin{equation*}
  {\bf i}_{X_{H}} \Omega_{\mathcal{M}} = {\bf d}H,
\end{equation*}
which gives 
\begin{equation}
  \label{eq:Heller}
  \begin{array}{c}
    \DS
    \dot{q} = \frac{p}{m},
    \qquad
    \dot{p} = -\exval{\nabla{V}},
    \qquad
    \dot{\mathcal{A}} = -\frac{1}{m}(\mathcal{A}^{2} - \mathcal{B}^{2}) - \exval{\nabla^{2}V},
    \qquad
    \dot{\mathcal{B}} = -\frac{1}{m}(\mathcal{A}\mathcal{B} + \mathcal{B}\mathcal{A}),
    \medskip\\
    \DS
    \dot{\phi} = \frac{p^{2}}{2m} - \exval{V}
    - \frac{\hbar}{2m} \tr\mathcal{B}
    + \frac{\hbar}{4} \tr\parentheses{ \mathcal{B}^{-1} \exval{\nabla^{2}V} },
    \qquad
    \dot{\delta} = \frac{\hbar}{2m} \tr\mathcal{A},
  \end{array}
\end{equation}
where $\nabla^{2}V$ is the $d \times d$ Hessian matrix, i.e.,
\begin{equation*}
  (\nabla^{2}V)_{ij} = \pd{^{2}V}{x_{i} \partial x_{j}},
\end{equation*}
and ${\exval{\,\cdot\,}}$ stands for the expectation value with respect to the normalized Gaussian $\psi_{0}$ defined in \eqref{eq:psi_0}, e.g., 
\begin{equation}
  \label{eq:avgV}
  \exval{V}(q, \mathcal{B}) \defeq \ip{ \psi_{0} }{ V\psi_{0} }
  = \sqrt{ \frac{\det \mathcal{B}}{(\pi\hbar)^{d}} } \int_{\R^{d}} V(x) \exp\brackets{ -\frac{1}{\hbar}(x - q)^{T}\mathcal{B}(x - q) } dx.
\end{equation}

However, the above system has an $\mathbb{S}^{1}$ phase symmetry, i.e., the Hamiltonian $H$ in \eqref{eq:H} does not depend on the phase $\phi$ and so is invariant under the phase shift in $\phi \in \mathbb{S}^{1}$.
Therefore, one may apply the Marsden--Weinstein reduction to reduce the system~\eqref{eq:Heller} to a lower-dimensional one.
The $\mathbb{S}^{1}$ phase symmetry gives rise to the momentum map
\begin{equation*}
  {\bf J}_{\!\mathcal{M}}(y) = -\hbar\,\norm{\chi(y;\,\cdot\,)}^{2} = -\hbar\,\sqrt{ \frac{(\pi\hbar)^{d}}{\det \mathcal{B}} }\, \exp\parentheses{ -\frac{2\delta}{\hbar} },
\end{equation*}
and the Marsden--Weinstein reduction~\cite{MaWe1974} (see also \citet[Sections~1.1 and 1.2]{MaMiOrPeRa2007}) yields the reduced space
\begin{equation*}
  {\mathcal{M}}_{\hbar} \defeq {\bf J}_{\!\mathcal{M}}^{-1}(-\hbar)/\mathbb{S}^{1} = T^{*}\R^{d} \times \Sigma_{d}
  = \{ (q, p, \mathcal{A}, \mathcal{B}) \},
\end{equation*}
with the reduced symplectic form
\begin{align}
  \label{eq:Omega-reduced}
  {\Omega}_{\hbar} \defeq -\d{\Theta}_{\hbar}
  &= \d{q_{i}} \wedge \d{p_{i}} + \frac{\hbar}{4}\,\mathcal{B}^{-1}_{jm} \mathcal{B}^{-1}_{nk} \d\mathcal{A}_{jk} \wedge \d\mathcal{B}_{mn}
  \nonumber\\
  &= \d{q_{i}} \wedge \d{p_{i}} + \frac{\hbar}{4}\, \d\mathcal{B}^{-1}_{jk} \wedge \d\mathcal{A}_{jk}.
\end{align}

\begin{remark}
  The second term in the above symplectic form is essentially the imaginary part of the Hermitian metric~\cite{Si1943}
  \begin{equation}
    \label{eq:metric-Sigma_d}
    g_{\Sigma_{d}} \defeq \tr\parentheses{ \mathcal{B}^{-1} \d\mathcal{C}\,\mathcal{B}^{-1} \d\bar{\mathcal{C}}\, }
    = \mathcal{B}^{-1}_{ik} \mathcal{B}^{-1}_{lj} \d\mathcal{C}_{kl} \otimes \d\bar{\mathcal{C}}_{ij}
  \end{equation}
  on the Siegel upper half space $\Sigma_{d}$, i.e.,
  \begin{equation}
    \label{eq:symplectic_form-Sigma_d}
    \Im g_{\Sigma_{d}} = -\mathcal{B}^{-1}_{ik} \mathcal{B}^{-1}_{lj} \d\mathcal{A}_{ij} \wedge \d\mathcal{B}_{kl},
  \end{equation}
  and this gives a symplectic structure on $\Sigma_{d}$.
\end{remark}

The corresponding Poisson bracket $\{\cdot,\cdot\}_{\hbar}$ is given as follows: For any $F, G \in C^{\infty}({\mathcal{M}}_{\hbar})$,
\begin{align}
  \label{eq:Poisson_bracket}
  \{F, G\}_{\hbar} &= {\Omega}_{\hbar}(X_{F}, X_{G})
  \nonumber\\
  &= \pd{F}{q_{i}} \pd{G}{p_{i}} - \pd{G}{q_{i}} \pd{F}{p_{i}}
  + \frac{4}{\hbar}\parentheses{
    \pd{F}{\mathcal{B}^{-1}_{jk}} \pd{G}{\mathcal{A}_{jk}} - \pd{G}{\mathcal{B}^{-1}_{jk}} \pd{F}{\mathcal{A}_{jk}}
  }.
\end{align}
Then the Gaussian wave packet dynamics~\eqref{eq:Heller} is reduced to the Hamiltonian system
\begin{equation}
  \label{eq:HamiltonianSystem-X_h}
  {\bf i}_{X_{H_{\hbar}}} {\Omega}_{\hbar} = {\bf d}H_{\hbar}
\end{equation}
with the reduced Hamiltonian
\begin{equation}
  \label{eq:H-reduced}
  H_{\hbar} = \frac{p^{2}}{2m} + \frac{\hbar}{4m}\tr\brackets{ \mathcal{B}^{-1}(\mathcal{A}^{2} + \mathcal{B}^{2}) } + \exval{V}(q, \mathcal{B}).
\end{equation}
Equivalently, in terms of the Poisson bracket~\eqref{eq:Poisson_bracket}, one has the system
\begin{equation*}
  \dot{w} = \{w, H_{\hbar}\}_{\hbar}
\end{equation*}
where
\begin{equation*}
  w \defeq (q, p, \mathcal{A}, \mathcal{B}) \in \mathcal{M}_{\hbar}.
\end{equation*}
As a result, \eqref{eq:HamiltonianSystem-X_h} gives the reduced set of semiclassical equations:
\begin{equation}
  \label{eq:Heller-reduced}
  \dot{q} = \frac{p}{m},
  \qquad
  \dot{p} = -{\exval{ \nabla{V} }},
  \qquad
  \dot{\mathcal{A}} = -\frac{1}{m}(\mathcal{A}^{2} - \mathcal{B}^{2}) - {\exval{ \nabla^{2}V }},
  \qquad
  \dot{\mathcal{B}} = -\frac{1}{m}(\mathcal{A}\mathcal{B} + \mathcal{B}\mathcal{A}).
\end{equation}
One obstacle in applying the above set of equations to practical problems is that the integral \eqref{eq:avgV} defining the potential term $\exval{V}$ in the Hamiltonian~\eqref{eq:H-reduced} cannot be evaluated explicitly unless the potential $V(x)$ takes fairly simple forms such as polynomials.
Therefore we evaluate ${\exval{V}}$ as an asymptotic expansion (see \cite[Section~7]{OhLe2013}) to find
\begin{equation}
  \label{eq:H-asymptotic-reduced}
  H_{\hbar} = H_{\hbar}^{1} + O(\hbar^{2})
  \quad\text{with}\quad
  H_{\hbar}^{1} \defeq \frac{p^{2}}{2m} + V(q)
    + \frac{\hbar}{4}\tr\brackets{ \mathcal{B}^{-1}\parentheses{ \frac{\mathcal{A}^{2} + \mathcal{B}^{2}}{m} + \nabla^{2}V(q) } }.
\end{equation}
Notice that the Hamiltonian is split into the classical one and a semiclassical correction proportional to $\hbar$.
Then the Hamiltonian system ${\bf i}_{X_{H_{\hbar}^{1}}} {\Omega}_{\mathcal{M}} = {\bf d}H_{\hbar}^{1}$ gives an asymptotic version of \eqref{eq:Heller-reduced}:
\begin{equation}
  \label{eq:Heller-asymptotic-reduced}
  \begin{array}{c}
    \DS
    \dot{q} = \frac{p}{m},
    \qquad
    \dot{p} = -\pd{}{q}\brackets{ V(q) + \frac{\hbar}{4} \tr\parentheses{ \mathcal{B}^{-1} \nabla^{2}V(q) } },
    \medskip\\
    \DS
    \dot{\mathcal{A}} = -\frac{1}{m}(\mathcal{A}^{2} - \mathcal{B}^{2}) - \nabla^{2}V(q),
    \qquad
    \dot{\mathcal{B}} = -\frac{1}{m}(\mathcal{A}\mathcal{B} + \mathcal{B}\mathcal{A}).
  \end{array}
\end{equation}
Writing $\mathcal{C} = \mathcal{A} + {\rm i}\mathcal{B}$, the last two equations above for $\mathcal{A}$ and $\mathcal{B}$ are combined into the following single Riccati-type equation:
\begin{equation}
  \label{eq:Riccati-C}
  \dot{\mathcal{C}} = -\frac{1}{m}\mathcal{C}^{2} - \nabla^{2}V(q).
\end{equation}
The main difference from the equations of \citet{He1975a,He1976b,He1981} and \citet{Ha1980,Ha1998} is that the equation for the momentum $p$ is not the classical one any more: It has a quantum correction term proportional to $\hbar$.
In fact, due to this correction term, the above set of equation~\eqref{eq:Heller-asymptotic-reduced} realizes semiclassical tunneling; see \citet[Section~9]{OhLe2013}.

\begin{remark}
  When $V(x)$ is quadratic, \eqref{eq:Heller} and \eqref{eq:Heller-asymptotic-reduced} recover the equations of \citet{He1975a,He1976b,He1981}; hence \eqref{eq:Heller} and \eqref{eq:Heller-asymptotic-reduced} are generalizations of the formulation of \citet{He1975a,He1976b,He1981} retaining its key property that it gives an exact solution to the Schr\"odinger equation~\eqref{eq:SchroedingerEq-L2} when $V(x)$ is quadratic.
  See also \citet[Section~7.1]{OhLe2013}.
\end{remark}

\section{Symmetry and Conservation Laws in Semiclassical Dynamics}
\label{sec:Symmetry_in_Semiclassical_Dynamics}
\subsection{Semiclassical Angular Momentum}
Suppose that the quantum mechanical system in question has a rotational symmetry, i.e., the potential $V\colon \R^{d} \to \R$ is invariant under the action of the Lie group $\SO(d)$ on the configuration space $\R^{d}$.
Does the semiclassical Hamiltonian system~\eqref{eq:Heller-reduced} or \eqref{eq:Heller-asymptotic-reduced} inherit the $\SO(d)$-symmetry?
If so, what kind of conservation laws follow from Noether's theorem?
The main result in this section answers these questions as follows:
\begin{theorem}
  \label{thm:semiclassical_angular_momentum}
  Let $\varphi\colon \SO(d) \times \R^{d} \to \R^{d}$ be the natural action of the rotation group $\SO(d)$ on the configuration space $\R^{d}$, i.e., for any $R \in \SO(d)$,
  \begin{equation}
    \label{eq:varphi}
    \varphi_{R}\colon \R^{d} \to \R^{d};
    \quad
    q \mapsto R q,
  \end{equation}
  and suppose that the potential $V \in C^{2}(\R^{d})$ is invariant under the $\SO(d)$-action $\varphi$ in \eqref{eq:varphi}, i.e.,
  \begin{equation}
    \label{eq:V-SO(d)_symmetry}
    V \circ \varphi_{R} = V.
  \end{equation}
  Then the quantity ${\bf J}_{\hbar}\colon {\mathcal{M}}_{\hbar} \to \so(d)^{*}$ defined by
  \begin{equation}
    \label{eq:semiclassical_angular_momentum}
    {\bf J}_{\hbar}(q, p, \mathcal{A}, \mathcal{B}) \defeq q \diamond p - \frac{\hbar}{2}[\mathcal{B}^{-1},\mathcal{A}],
  \end{equation}
  with $q \diamond p$ denoting (see, e.g., \citet[Remark~6.3.3 on p.~150]{Ho2011b})
  \begin{equation*}
    (q \diamond p)_{ij} \defeq q_{j}p_{i} - q_{i}p_{j},
  \end{equation*}
  is conserved along the solutions of the semiclassical system~\eqref{eq:Heller-reduced} or \eqref{eq:Heller-asymptotic-reduced}.
\end{theorem}

Note that setting $\hbar = 0$ in \eqref{eq:semiclassical_angular_momentum} recovers the classical angular momentum ${\bf J}_{0} = q \diamond p$; so ${\bf J}_{\hbar}$ is considered to be a semiclassical extension of the angular momentum; so we call ${\bf J}_{\hbar}$ the {\em semiclassical angular momentum}.


The first step towards the proof of Theorem~\ref{thm:semiclassical_angular_momentum} is to identify the natural $\SO(d)$-action on the symplectic manifold ${\mathcal{M}}_{\hbar} = T^{*}\R^{d} \times \Sigma_{d}$ induced by the action $\varphi$ on $\R^{d}$ defined above in \eqref{eq:varphi}.
For the cotangent bundle component $T^{*}\R^{d}$, the natural choice is the cotangent lift $\Phi\colon \SO(d) \times T^{*}\R^{d} \to T^{*}\R^{d}$ defined by
\begin{equation}
  \label{eq:Phi-SO(d)}
  \Phi_{R} \defeq T^{*}\varphi_{R^{-1}}\colon T^{*}\R^{d} \to T^{*}\R^{d}; 
  \qquad
  (q, p) \mapsto (R q, R p).
\end{equation}

What is the natural induced action on $\Sigma_{d}$ induced by $\varphi$?
We would like to define the action so that the Gaussian wave packet quantum state $[\psi_{0}]_{\mathbb{S}^{1}}$ defined in \eqref{eq:psi_0-nophase} is invariant under the action on all of its variables $(x,q,p,\mathcal{A},\mathcal{B})$.
Notice first that the natural $\SO(d)$-action on the variables $(x, q, p)$ is, for any $R \in \SO(d)$,
\begin{equation*}
  (x, q, p) \mapsto (R x, R q, R p).
\end{equation*}
In order for the state $[\psi_{0}]_{\mathbb{S}^{1}}$ to be invariant under the action of $\SO(d)$, one may define an action of $\SO(d)$ on $\mathcal{A} + {\rm i}\mathcal{B}$ as follows:
\begin{equation}
  \label{eq:SO(d)-action_Sigma_d}
  \gamma_{R}\colon \Sigma_{d} \to \Sigma_{d};
  \quad
  \mathcal{A} + {\rm i}\mathcal{B} \mapsto R(\mathcal{A} + {\rm i}\mathcal{B})R^{T}.
\end{equation}
As a result, we have a natural $\SO(d)$-action on the symplectic manifold ${\mathcal{M}}_{\hbar} = T^{*}\R^{d} \times \Sigma_{d}$:

\begin{proposition}
  \label{prop:SO(d)-action}
  Define an $\SO(d)$-action $\Gamma\colon \SO(d) \times {\mathcal{M}}_{\hbar} \to {\mathcal{M}}_{\hbar}$ by
  \begin{equation}
    \label{eq:Gamma}
    \Gamma_{R}\colon {\mathcal{M}}_{\hbar} \to {\mathcal{M}}_{\hbar};
    \qquad
    (q, p, \mathcal{A}, \mathcal{B}) \mapsto (R q, R p, R\mathcal{A}R^{T}, R\mathcal{B}R^{T})
  \end{equation}
  for any $R \in \SO(d)$. Then:
  \begin{enumerate}[(i)]
  \item The Gaussian wave packet state $[\psi_{0}]_{\mathbb{S}^{1}}$ in \eqref{eq:psi_0-nophase} is invariant under the action
    \begin{align*}
      \varphi_{R} \times \Gamma_{R}\colon & \R^{d} \times {\mathcal{M}}_{\hbar} \to \R^{d} \times {\mathcal{M}}_{\hbar};
      \\
      & (x, q, p, \mathcal{A}, \mathcal{B}) \mapsto (R x, R q, R p, R\mathcal{A}R^{T}, R\mathcal{B}R^{T}).
    \end{align*}
  \item The action $\Gamma$ is symplectic with respect to the symplectic form ${\Omega}_{\hbar}$ defined in \eqref{eq:Omega-reduced}, i.e., $\Gamma_{R}^{*} {\Omega}_{\hbar} = {\Omega}_{\hbar}$ for any $R \in \SO(d)$.
  \end{enumerate}
\end{proposition}

\begin{proof}
  Follows from simple calculations.
\end{proof}

It is now easy to see that the $\SO(d)$-symmetry~\eqref{eq:V-SO(d)_symmetry} of the potential $V$ implies that of the semiclassical Hamiltonian~\eqref{eq:H-reduced} or \eqref{eq:H-asymptotic-reduced} in the following sense:
\begin{lemma}
  \label{lemma:SO(d)-symmetry}
  If the potential $V \in C^{2}(\R^{d})$ is invariant under the $\SO(d)$-action $\varphi$ in \eqref{eq:varphi}, i.e., \eqref{eq:V-SO(d)_symmetry} is satisfied, then the Hamiltonians $H_{\hbar}$ and $H_{\hbar}^{1}$, \eqref{eq:H-reduced} and \eqref{eq:H-asymptotic-reduced} respectively, are invariant under the $\SO(d)$-action~\eqref{eq:Gamma}, i.e.,
  \begin{equation*}
    H_{\hbar} \circ \Gamma_{R} = H_{\hbar},
    \qquad
    H_{\hbar}^{1} \circ \Gamma_{R} = H_{\hbar}^{1}.
  \end{equation*}
\end{lemma}

\begin{proof}
  The invariance of the terms that do not involve the potential $V$ follows from simple calculations.
  For those terms with the potential $V$, we first have
  \begin{align*}
    {\exval{V}}(R q, R\mathcal{B}R^{T})
    &= \sqrt{ \frac{\det \mathcal{B}}{(\pi\hbar)^{d}} } \int_{\R^{d}} V(x) \exp\brackets{ -\frac{1}{\hbar}(x - R q)^{T}R\mathcal{B}R^{T}(x - R q) } dx
    \\
    &= \sqrt{ \frac{\det \mathcal{B}}{(\pi\hbar)^{d}} } \int_{\R^{d}} \underbrace{V(R\xi)}_{V(\xi)} \exp\brackets{ -\frac{1}{\hbar}(\xi - q)^{T}\mathcal{B}(\xi - q) } d\xi
    \\
    &= {\exval{V}}(q, \mathcal{B}),
  \end{align*}
  where we set $\xi = R^{T} x$.
  The invariance of the Laplacian $\nabla^{2}V$ follows from the $\SO(d)$-invariance of the Laplacian and the potential $V$ itself:
  \begin{equation*} 
    (\nabla^{2} V) \circ \varphi_{R} = \nabla^{2}(V \circ \varphi_{R}) = \nabla^{2}V. \qedhere
  \end{equation*}
\end{proof}

Now we are ready to prove Theorem~\ref{thm:semiclassical_angular_momentum}:
\begin{proof}[Proof of Theorem~\ref{thm:semiclassical_angular_momentum}]
  We show that the semiclassical angular momentum~\eqref{eq:semiclassical_angular_momentum} is in fact the momentum map corresponding to the action $\Gamma$ defined in \eqref{eq:Gamma}.
  Then the result follows from Noether's theorem (see, e.g., \citet[Theorem~11.4.1 on p.~372]{MaRa1999}), because Proposition~\ref{prop:SO(d)-action} guarantees that $\Gamma$ is a symplectic $\SO(d)$-action on the symplectic manifold ${\mathcal{M}}_{\hbar} = T^{*}\R^{d} \times \Sigma_{d}$, and the assumption~\eqref{eq:V-SO(d)_symmetry} on the $\SO(d)$-symmetry of the potential $V$ along with Lemma~\ref{lemma:SO(d)-symmetry} implies that the Hamiltonian $H_{\hbar}$ and $H_{\hbar}^{1}$ are both invariant under the action.

  Let $\xi$ be an arbitrary element in the Lie algebra $\so(d)$.
  Then the corresponding infinitesimal generator is then given by
  \begin{equation*}
    \xi_{{\mathcal{M}}_{\hbar}}(w)
    \defeq \left. \od{}{\eps} \Gamma_{\exp(\eps\xi)}(w) \right|_{\eps=0}
    = \xi q \cdot \pd{}{q}
    + \xi p \cdot \pd{}{p}
    + [\xi,\mathcal{A}]_{jk} \pd{}{\mathcal{A}_{jk}}
    + [\xi,\mathcal{B}]_{jk} \pd{}{\mathcal{B}_{jk}}.
  \end{equation*}
  Let us equip $\so(d)$ with the inner product $\ip{\cdot}{\cdot}$ defined as
  \begin{equation*}
    \ip{\cdot}{\cdot}\colon \so(d) \times \so(d) \to \R;
    \qquad
    (\xi, \eta) \mapsto \ip{\xi}{\eta} \defeq \frac{1}{2}\tr(\xi^{T}\eta).
  \end{equation*}
  Note that the dual $\so(d)^{*}$ of the Lie algebra $\so(d)$ may be identified with $\so(d)$ itself via the inner product.
  We may then define $J_{\hbar}(\xi)\colon {\mathcal{M}}_{\hbar} \to \R$ for each $\xi \in \so(d)$ as
  \begin{align*}
    J_{\hbar}(\xi)(w)
    &\defeq \ip{ \Theta_{{\mathcal{M}}_{\hbar}}(w) }{ \xi_{{\mathcal{M}}_{\hbar}}(w) }
    \\
    &= p^{T}\xi q - \frac{\hbar}{4}\tr(\mathcal{B}^{-1}[\xi,\mathcal{A}])
    \\
    &= \ip{ q \diamond p - \frac{\hbar}{2}[\mathcal{B}^{-1},\mathcal{A}]}{\xi},
  \end{align*}
  We also used the following identity: For any $\xi, \eta, \zeta \in \so(d)$,
  \begin{equation*}
    \ip{\xi}{[\eta,\zeta]} = \ip{\eta}{[\zeta,\xi]}.
  \end{equation*}
  Then the corresponding momentum map ${\bf J}_{\hbar}\colon {\mathcal{M}}_{\hbar} \to \so(d)^{*}$ is defined so that, for any $\xi \in \so(d)$,
  \begin{equation}
    \label{eq:J-so(d)}
    J_{\hbar}(\xi)(w) = \ip{ {\bf J}_{\hbar}(w) }{\xi},
  \end{equation}
  which gives the semiclassical angular momentum~\eqref{eq:semiclassical_angular_momentum}.
\end{proof}

\begin{remark}
  In particular, if $d = 3$, we may identify $\so(3)$ with $\R^{3}$ by the ``hat map'' (see, e.g., \citet[p.~289]{MaRa1999})
  \begin{equation}
    \label{eq:hat}
    \hat{(\,\cdot\,)}\colon \R^{3} \to \so(3);
    \qquad
    v = (v_{1}, v_{2}, v_{3}) \mapsto \hat{v} \defeq 
    \begin{bmatrix}
      0 & -v_{3} & v_{2} \\
      v_{3} & 0 & -v_{1} \\
      -v_{2} & v_{1} & 0
    \end{bmatrix}
  \end{equation}
  and write its inverse as $\vee\colon \so(3) \to \R^{3}$; then $\widehat{q \times p} = q \diamond p$ or equivalently $(q \diamond p)^{\vee} = q \times p$.  As a result, we have
  \begin{equation*}
    \ip{ {\bf J}_{\hbar}(w) }{\xi}
    = \parentheses{ q \times p - \frac{\hbar}{2}[\mathcal{B}^{-1},\mathcal{A}]^{\vee} } \cdot \xi^{\vee}
  \end{equation*}
  and thus may define the 3-dimensional semiclassical angular momentum vector $\vec{J}_{\hbar}\colon {\mathcal{M}}_{\hbar} \to \R^{3}$ as follows:
  \begin{equation}
    \label{eq:semiclassical_angular_momentum-3d}
    \vec{J}_{\hbar}(w) \defeq ({\bf J}_{\hbar}(w))^{\vee} = q \times p - \frac{\hbar}{2}[\mathcal{B}^{-1},\mathcal{A}]^{\vee}.
  \end{equation}
  So the semiclassical angular momentum vector $\vec{J}_{\hbar}$ is the classical angular momentum vector $q \times p$ plus an additional quantum term proportional to $\hbar$.
\end{remark}

\subsection{Properties of the Semiclassical Angular Momentum}
The semiclassical angular momentum ${\bf J}_{\hbar}\colon {\mathcal{M}}_{\hbar} \to \so(d)^{*}$ retains the main features of the classical angular momentum due to its geometrically natural construction.
Particularly, we have the following:
\begin{proposition}
  The semiclassical angular momentum ${\bf J}_{\hbar}\colon {\mathcal{M}}_{\hbar} \to \so(d)^{*}$ defined in \eqref{eq:semiclassical_angular_momentum} is an equivariant momentum map under the $\SO(d)$-action, and its components satisfy the equality
  \begin{equation}
    \label{eq:J-Poisson_bracket}
    \left\{ J_{\hbar}^{jk}, J_{\hbar}^{rs} \right\}_{\hbar}
    = \delta_{kr}\, J_{\hbar}^{js} - \delta_{ks}\, J_{\hbar}^{jr}
    + \delta_{js}\, J_{\hbar}^{kr} - \delta_{jr}\, J_{\hbar}^{ks},
  \end{equation}
  with respect to the semiclassical Poisson bracket~\eqref{eq:Poisson_bracket}, where $j, k, r, s \in \{1, \dots, d\}$.
  Particularly, when $d = 3$, the components of the semiclassical angular momentum vector $\vec{J}_{\hbar}\colon {\mathcal{M}}_{\hbar} \to \R^{3}$ defined in \eqref{eq:semiclassical_angular_momentum-3d} satisfy, for any $i,j,k \in \{1, 2, 3\}$ such that $\epsilon_{ijk} = 1$,
  \begin{equation}
    \label{eq:J-Poisson_bracket-3d}
    \left\{ (\vec{J}_{\hbar})_{i}, (\vec{J}_{\hbar})_{j} \right\}_{\hbar}
    = (\vec{J}_{\hbar})_{k}.
  \end{equation}
  That is, the semiclassical angular momentum \eqref{eq:semiclassical_angular_momentum} along with the semiclassical Poisson bracket~\eqref{eq:Poisson_bracket} is a natural extension of the classical angular momentum.
\end{proposition}

\begin{proof}
  Let $\Ad_{R^{-1}}^{*}\colon \so(d)^{*} \to \so(d)^{*}$ be the coadjoint action of $\SO(d)$ on $\so(d)$, i.e.,
  \begin{equation*}
    \Ad^{*}_{{R^{-1}}} \mu = R \mu R^{T}.
  \end{equation*}
  Then it is straightforward to check that the semiclassical angular momentum ${\bf J}_{\hbar}\colon {\mathcal{M}}_{\hbar} \to \so(d)^{*}$ defined in \eqref{eq:semiclassical_angular_momentum} is an equivariant momentum map, i.e., 
  \begin{equation}
    \label{eq:equivariance}
    {\bf J}_{\hbar} \circ \Gamma_{R} = \Ad^{*}_{{R^{-1}}} \circ\, {\bf J}_{\hbar},
  \end{equation}
  or more concretely,
  \begin{equation*}
    {\bf J}_{\hbar}(R q, R p, R\mathcal{A}R^{T}, R\mathcal{B}R^{T}) = R\, {\bf J}_{\hbar}(q, p, \mathcal{A}, \mathcal{B}) R^{T},
  \end{equation*}
  just like the classical angular momentum ${\bf J}_{0}(q, p) \defeq q \diamond p$:
  \begin{equation*}
    {\bf J}_{0}(R q, R p) = R\, {\bf J}_{0}(q, p) R^{T}.
  \end{equation*}
  Furthermore, the equivariance \eqref{eq:equivariance} implies that (see, e.g., \cite[Corollary~4.2.9 on p.~281]{AbMa1978}), for any $\xi, \eta \in \so(d)$, 
  \begin{equation*}
    \left\{ J_{\hbar}(\xi), J_{\hbar}(\eta) \right\}_{\hbar} = J_{\hbar}([\xi,\eta]),
  \end{equation*}
  where $\{\cdot,\cdot\}_{\hbar}$ is the Poisson bracket defined in \eqref{eq:Poisson_bracket}.
  Now let $E_{ij} \defeq e_{i}e_{j}^{T} - e_{j}e_{i}^{T}$, where $e_{i} \in \R^{d}$ with $i \in \{1, \dots, d\}$ is the unit vector whose $i$-th entry is 1.
  Clearly $E_{ij} \in \so(d)$ for any $i, j \in \{1, \dots, d\}$, and for any $A \in \so(d)^{*} \cong \so(d)$, we have $\ip{A}{E_{ij}} = A_{ij}$, and so \eqref{eq:J-so(d)} gives
  \begin{equation*}
    J_{\hbar}(E_{jk}) = J_{\hbar}^{jk}.
  \end{equation*}
  They also satisfy the identity
  \begin{equation*}
    [E_{jk}, E_{rs}] = \delta_{kr} E_{js} - \delta_{ks} E_{jr} + \delta_{js} E_{kr} - \delta_{jr} E_{ks}.
  \end{equation*}
  Therefore we have
  \begin{align*}
    \left\{ J_{\hbar}(E_{jk}), J_{\hbar}(E_{rs}) \right\}_{\hbar}
    &= J_{\hbar}([E_{jk}, E_{rs}])
    \\
    &= \delta_{kr}\, J_{\hbar}(E_{js}) - \delta_{ks}\, J_{\hbar}(E_{jr})
    + \delta_{js}\, J_{\hbar}(E_{kr}) - \delta_{jr}\, J_{\hbar}(E_{ks}),
  \end{align*}
  which gives \eqref{eq:J-Poisson_bracket}.

  In particular, for $d = 3$, let $i,j,k \in \{1, 2, 3\}$ such that $\epsilon_{ijk} = 1$, we have, without assuming summation on $k$,
  \begin{equation*}
    \left\{ J_{\hbar}^{jk}, J_{\hbar}^{ki} \right\}_{\hbar}
    = J_{\hbar}^{ji}
    = -J_{\hbar}^{ij}.
  \end{equation*}
  Notice however $J_{\hbar}^{jk} = -(\vec{J}_{\hbar})_{i}$ etc., and so the components of $\vec{J}_{\hbar}$ satisfy the relationship~\eqref{eq:J-Poisson_bracket-3d} again just like the classical angular momentum $\vec{J}_{0}$ does with the classical Poisson bracket.
\end{proof}

The semiclassical angular momentum is compatible with the quantum picture as well:
Let ${\bf x}_{\text{op}}$ and ${\bf p}_{\text{op}}$ be the position and momentum operators of the canonical quantization, i.e., ${\bf x}_{\text{op}}$ is the multiplication by the position vector $x$, whereas ${\bf p}_{\text{op}} \defeq -i \hbar \nabla$.
Then, it is a tedious but straightforward calculation to show that
\begin{equation*}
  \exval{{\bf x}_{\text{op}} \times {\bf p}_{\text{op}}} =
  \ip{\psi_{0}}{({\bf x}_{\text{op}} \times {\bf p}_{\text{op}}) \psi_{0}} = \vec{J}_{\hbar},
\end{equation*}
that is, the expectation value of the angular momentum operator ${\bf x}_{\text{op}} \times {\bf p}_{\text{op}}$ with respect to the normalized Gaussian wave packet~\eqref{eq:psi_0} coincides with the semiclassical angular momentum~\eqref{eq:semiclassical_angular_momentum-3d}.

\subsection{Example and Numerical Results}
We would like to illustrate the conservation of the semiclassical angular momentum~\eqref{eq:semiclassical_angular_momentum} in the following simple two-dimensional example with rotational symmetry.
\begin{example}[Two-dimensional quartic potential]
  Let $m = 1$ and $\hbar = 0.005$, and consider the two-dimensional case, i.e., $d = 2$, with the following quartic potential with $\SO(2)$-symmetry:
  \begin{equation*}
    V(x_{1}, x_{2}) = V_{R}(\abs{x})
    \quad
    \text{where}
    \quad
    V_{R}(r) = \frac{r^{2}}{2} + \frac{r^{4}}{4}.
  \end{equation*}
  The initial condition is chosen as follows:
  \begin{equation*}
    q(0) = (1,0),
    \qquad
    p(0) = (0, 1),
    \qquad
    \mathcal{A}(0) + {\rm i}\mathcal{B}(0) =
    \begin{bmatrix}
      1 + {\rm i} & 0.5(1+{\rm i}) \\
      0.5(1+{\rm i}) & 1 + {\rm i}
    \end{bmatrix}.
  \end{equation*}
  We used the variational splitting method of \citet[Chapter~IV.4]{Lu2008} (see also \citet{FaLu2006}) for the semiclassical equations~\eqref{eq:Heller-asymptotic-reduced}; it is easy to see that the integrator preserves the symplectic structure ${\Omega}_{\hbar}$.
  The corresponding classical solution is obtained by the St\"ormer--Verlet method\footnote{The variational splitting integrator is a natural extension of the St\"ormer--Verlet method in the sense that it recovers the St\"ormer--Verlet method as $\hbar \to 0$~\cite{FaLu2006}.}, and the time step is $\Delta t = 0.01$ for both solutions.
  \begin{figure}[hbtp]
    \centering
    \subfigure[Classical and semiclassical orbits for $0 \le t \le 50$.]{
      \includegraphics[width=.425\linewidth]{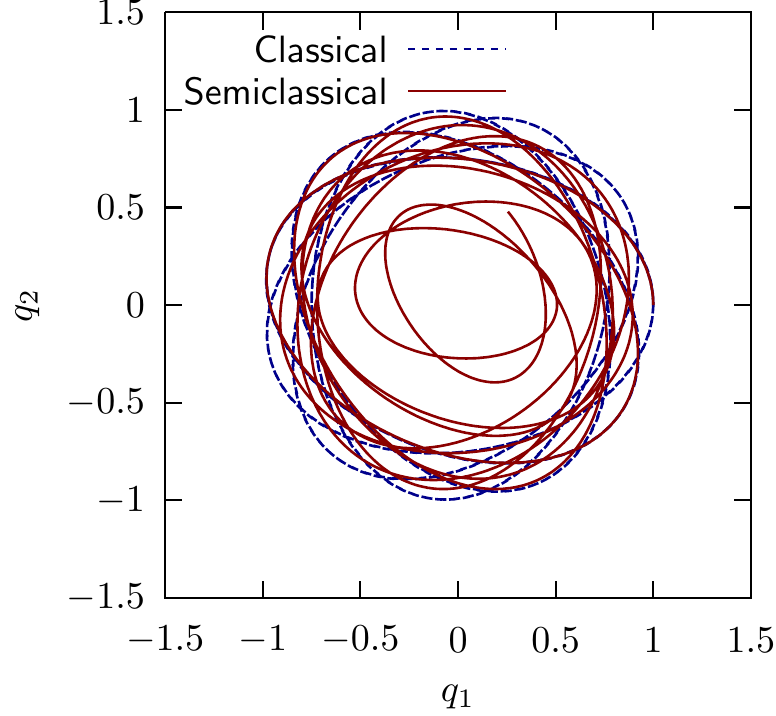}
    }
    \quad
    \subfigure[Time evolution of the classical and semiclassical angular momenta ${\bf J}_{0}$ and ${\bf J}_{\hbar}$, {\em both along the semiclassical dynamics~\eqref{eq:Heller-asymptotic-reduced}}.]{
      \includegraphics[width=.425\linewidth]{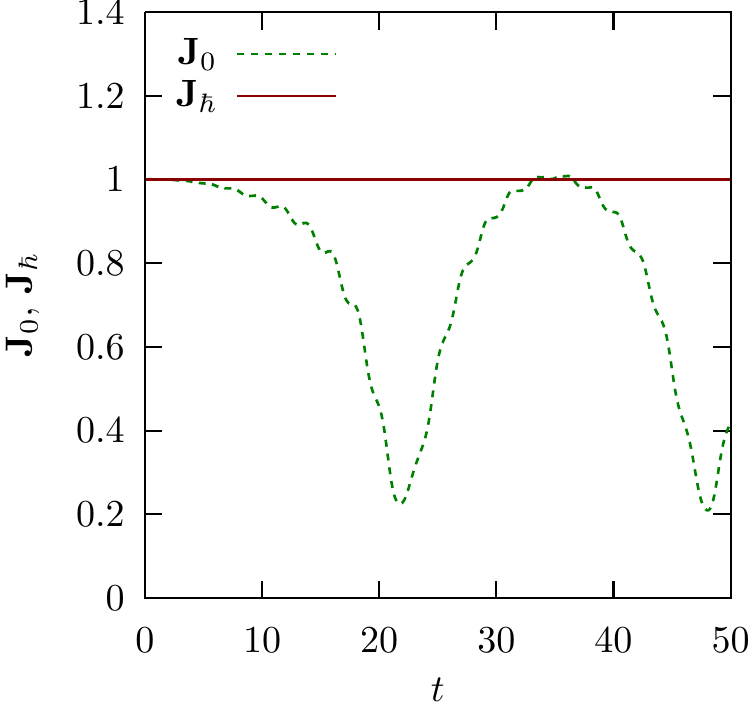}
    }
    \caption{The semiclassical solutions under axisymmetric quartic potential.}
    \label{fig:2D_Quartic}
  \end{figure}

  Figure~\ref{fig:2D_Quartic} shows the classical and semiclassical orbits for $0 \le t \le 50$ as well as the time evolution of the classical angular momentum ${\bf J}_{0} = q_{1} p_{2} - p_{1} q_{2}$ and the semiclassical one ${\bf J}_{\hbar}$ (the third component of \eqref{eq:semiclassical_angular_momentum}) {\em both along the solution of the semiclassical equations~\eqref{eq:Heller-asymptotic-reduced}}.
  One sees that the semiclassical angular momentum ${\bf J}_{\hbar}$ is conserved, whereas the classical one ${\bf J}_{0}$ fluctuates significantly.
\end{example}

\section{Symmetry and Conservation Laws in the Hagedorn Wave Packet Dynamics}
\label{sec:Symmetry_in_Hagedorn}
So far we have been looking into symmetry and conservation laws based on the symplectic formulation of the Gaussian wave packet dynamics from Section~\ref{sec:Symplectic_Semiclassical_Dynamics}.
In this section, we change our focus to the more prevalent formulation by \citet{Ha1980, Ha1998}.
This formulation leads to the elegant derivation by \citet{Ha1998} of raising and lowering operators for the Gaussian wave packet and also an orthonormal basis generated from them.
\citet{FaGrLu2009} exploited this orthonormal basis to develop a numerical method to solve the semiclassical Schr\"odinger equation more efficiently than with the Fourier basis.

\subsection{The Hagedorn Wave Packet Dynamics}
\label{ssec:Hagedorn_dynamics}
\citet{Ha1980} (see also \citet{Heller-LesHouches} and \citet[Chapter~V]{Lu2008}) used a slightly different parametrization of the Gaussian wave packet \eqref{eq:chi}.
More precisely, the elements $\mathcal{C} = \mathcal{A} + {\rm i}\mathcal{B}$ in the Siegel upper half space $\Sigma_{d}$ are parametrized as $\mathcal{C} = P Q^{-1}$ with $Q$ and $P$ being $d \times d$ complex matrices that satisfy
\begin{equation*}
  Q^{T}P - P^{T}Q = 0
  \quad\text{and}\quad
  Q^{*}P - P^{*}Q = 2{\rm i}I_d.
\end{equation*}
So the Gaussian wave packet~\eqref{eq:chi} can now be written as
\begin{equation*}
  \chi(y;x) = \exp\braces{ \frac{{\rm i}}{\hbar}\brackets{ \frac{1}{2}(x - q)^{T}P Q^{-1}(x - q) + p \cdot (x - q) + (\phi + {\rm i}\delta) } },
\end{equation*}
where $y \defeq (q, p, Q, P, \phi, \delta)$.
Its norm is then
\begin{equation*}
  \mathcal{N}(Q,\delta) \defeq \norm{\chi(y;\cdot)}^{2} =  (\pi\hbar)^{d/2}\, |\det Q| \exp\parentheses{ -\frac{2\delta}{\hbar} }.
\end{equation*}
Hence the wave packet is normalized as follows:
\begin{align}
  \label{eq:chi-Hagedorn-normalized}
  \frac{\chi(y;x)}{\norm{\chi(y;\cdot)}}
  &= (\pi\hbar)^{-d/4} |\det Q|^{-1/2} \exp\braces{ \frac{{\rm i}}{\hbar}\brackets{ \frac{1}{2}(x - q)^{T}P Q^{-1}(x - q) + p \cdot (x - q) + \phi } }
  \nonumber\\
  &= e^{{\rm i}S/\hbar}\, \varphi_{0}(q, p, Q, P; x),
\end{align}
where we defined the new variable
\begin{equation*}
  S \defeq \phi - \frac{\hbar}{2}\arg(\det Q)
\end{equation*}
and the ``ground state'' $\varphi_{0}$ of the Hagedorn wave packets
\begin{equation*}
  \varphi_{0}(q, p, Q, P; x) \defeq (\pi\hbar)^{-d/4} (\det Q)^{-1/2} \exp\braces{ \frac{{\rm i}}{\hbar}\brackets{ \frac{1}{2}(x - q)^{T}P Q^{-1}(x - q) + p \cdot (x - q) } },
\end{equation*}
where an appropriate branch cut is taken for $(\det Q)^{1/2}$.

\citet{Ha1980, Ha1998} showed that \eqref{eq:chi-Hagedorn-normalized} is an exact solution of the Schr\"odinger equation if the potential $V(x)$ is quadratic and also the parameters $(q, p, Q, P)$ satisfy
\begin{equation}
  \label{eq:Hagedorn}
  \dot{q} = \frac{p}{m},
  \qquad
  \dot{p} = -\nabla V(q),
  \qquad
  \dot{Q} = \frac{P}{m},
  \qquad
  \dot{P} =  -\nabla^{2}V(q)\,Q,
\end{equation}
and the quantity $S(t)$ is the classical action integral evaluated along the solution $(q(t),p(t))$, i.e.,
\begin{equation*}
  S(t) = S(0) + \int_{0}^{t} \parentheses{ \frac{p(s)^{2}}{2m} - V(q(s)) }\,ds.
\end{equation*}
\citet{Ha1980, Ha1998} also showed that \eqref{eq:chi-Hagedorn-normalized} gives an $O(t\sqrt{\hbar})$ approximation even when the potential $V(x)$ is not quadratic as long as it satisfies some regularity assumptions.

\subsection{The Hagedorn Parametrization and Symplectic Group $\Sp(2d,\R)$}
As pointed out by \citet[Section~V.1]{Lu2008}, the matrices $Q$ and $P$ used above to parametrize the Siegel upper half space $\Sigma_{d}$ constitute a symplectic matrix of degree $d$.
Specifically, \citet[Section~V.1]{Lu2008} shows that the symplectic group $\Sp(2d,\R)$ is written as follows:
\begin{align*}
  \Sp(2d,\R)
  &= \setdef{
    \begin{bmatrix}
      \Re Q & \Im Q \smallskip\\
      \Re P & \Im P
    \end{bmatrix}
  }{Q, P \in {\sf M}_{d}(\C),\, Q^{T}P - P^{T}Q = 0,\, Q^{*}P - P^{*}Q = 2{\rm i}I_d}
  \\
  &= \setdef{
    (Q, P) \in {\sf M}_{d}(\C) \times {\sf M}_{d}(\C)
  }{Q^{T}P - P^{T}Q = 0,\, Q^{*}P - P^{*}Q = 2{\rm i}I_d},
\end{align*}
where ${\sf M}_{d}(\C)$ stands for the set of all $d \times d$ complex matrices; it is also shown that Hagedorn's parametrization of $\Sigma_{d}$ is nothing but the explicit description of the map $\pi_{\U(d)}\colon \Sp(2d,\R) \to \Sigma_{d}$ given by
\begin{equation}
  \label{eq:pi_Ud-2}
  \pi_{\U(d)}\colon
  \begin{bmatrix}
    \Re Q & \Im Q \smallskip\\
    \Re P & \Im P
  \end{bmatrix} 
  \text{ or }
  (Q, P)
  \mapsto
  P Q^{-1}.
\end{equation}
As shown in Appendix~\ref{sec:Sigma_d}, this is the natural quotient map that comes from the fact that the Siegel upper half space $\Sigma_{d}$ is identified as the homogeneous space $\Sp(2d,\R)/\U(d)$.

By setting 
\begin{equation*}
  Y(t) =
  \begin{bmatrix}
    \Re Q(t) & \Im Q(t) \smallskip\\
    \Re P(t) & \Im P(t)
  \end{bmatrix},
\end{equation*}
we can rewrite the last two equations in \eqref{eq:Hagedorn} for $Q$ and $P$ as 
\begin{equation}
  \label{eq:Hagedorn-Sp}
  \dot{Y}(t) = \xi(t)\,Y(t),
\end{equation}
where $\xi(t)$ is defined as
\begin{equation}
  \label{eq:Hagedorn-xi}
  \xi(t) \defeq 
  \begin{bmatrix}
    0 & I_{d}/m \smallskip\\
    -\nabla^{2}V(q(t)) & 0
  \end{bmatrix},
\end{equation}
and is an element in the Lie algebra $\mathfrak{sp}(2d,\R)$ of $\Sp(2d,\R)$.

Now, defining a curve $Y(t)$ in $\Sp(2d,\R)$ by \eqref{eq:Hagedorn-Sp}, the equations~\eqref{eq:Hagedorn} of Hagedorn may be rewritten as
\begin{equation}
  \label{eq:Hagedorn2}
  \dot{q} = \frac{p}{m},
  \qquad
  \dot{p} = -\nabla V(q),
  \qquad
  \dot{Y} = \xi Y,
\end{equation}
and hence defines a time evolution in $T^{*}\R^{d} \times \Sp(2d,\R)$.
Note that the dimension of $T^{*}\R^{d} \times \Sp(2d,\R)$ is $2d^{2} + 3d = d(2d+3)$, which is odd if $d$ is odd, and so $T^{*}\R^{d} \times \Sp(2d,\R)$ cannot be a symplectic manifold when $d$ is odd.

\begin{remark}
  This result suggests that the Hagedorn wave packet dynamics is a lift of the $\Sigma_{d}$-component of the semiclassical dynamics~\eqref{eq:Heller-reduced} or \eqref{eq:Heller-asymptotic-reduced} to the symplectic group $\Sp(2d,\R)$; see Proposition~\ref{prop:Hagedorn-Heller} in Appendix~\ref{sec:Sigma_d} for this connection between the two formulations.
\end{remark}

\subsection{Symmetry and Conservation Laws in the Hagedorn Wave Packet Dynamics}
Suppose that the classical Hamiltonian system
\begin{equation}
  \label{eq:classicalHamSys}
  \dot{q} = \pd{H}{p},
  \qquad
  \dot{p} = -\pd{H}{q}
\end{equation}
has a Lie group symmetry and therefore, by Noether's theorem, possesses a conserved quantity.
Do the semiclassical equations~\eqref{eq:Hagedorn} of Hagedorn inherit the symmetry and the conservation law?
Note that Noether's theorem does not directly extend to \eqref{eq:Hagedorn} because, as mentioned above, \eqref{eq:Hagedorn} is not defined on a symplectic manifold.

In what follows, we will give an affirmative answer to the above question by proving the following:
\begin{theorem}
  \label{thm:Noether_for_Hagedorn}
  Suppose that the classical Hamiltonian system~\eqref{eq:classicalHamSys} on $T^{*}\R^{d}$ has a symmetry under the action of a Lie group $G$, that is, let $\phi\colon G \times \R^{d} \to \R^{d}$ be the action of $G$ on the configuration space $\R^{d}$, and suppose that the Hamiltonian $H\colon T^{*}\R^{d} \to \R$ has a $G$-symmetry under the action of the cotangent lift $\Phi \defeq T^{*}\phi\colon G \times T^{*}\R^{d} \to T^{*}\R^{d}$, i.e., $H \circ \Phi_{g} = H$.
  Let ${\bf J}\colon T^{*}\R^{d} \to \mathfrak{g}^{*}$ be the corresponding momentum map.
  Then the quantity
  \begin{equation}
    \label{eq:ConservedQuantity-Hagedorn}
    (D{\bf J}(z) \cdot Y)_{k} = \pd{{\bf J}}{z^{j}}(z)\, Y_{jk}
  \end{equation}
  for $k \in \{1, \dots, d\}$ is conserved along the solutions $(z(t), Y(t)) \in T^{*}\R^{d} \times \Sp(2d,\R)$ of the equation~\eqref{eq:Hagedorn2} of Hagedorn.
\end{theorem}

Before proving this theorem, let us work out an interesting special case:
\begin{example}[$\SO(3)$-symmetry]
  \label{ex:Noether_for_Hagedorn-SO3}
  Let $d = 3$ and $G = \SO(3)$, i.e., the classical Hamiltonian system~\eqref{eq:classicalHamSys} in $T^{*}\R^{3}$ has a rotational symmetry.
  The corresponding momentum map ${\bf J}\colon T^{*}\R^{3} \to \so(3)^{*} \cong \R^{3}$ is given by ${\bf J} = q \times p$.
  Then we have the $3 \times 6$ matrix
  \begin{equation*}
    D{\bf J}(z) = 
    \brackets{
      -\hat{p}
      \;|\;
      \hat{q}
    },
  \end{equation*}
  where we used the hat map $\hat{(\,\cdot\,)}\colon \R^{3} \to \so(3)$ defined in \eqref{eq:hat}.
  Therefore the conserved quantity~\eqref{eq:ConservedQuantity-Hagedorn} is given by
  \begin{equation*}
    D{\bf J}(z) \cdot Y
    = \brackets{
      \hat{q}\,\Re P - \hat{p}\,\Re Q
      \;|\;
      \hat{q}\,\Im P - \hat{p}\,\Im Q
      },
  \end{equation*}
  that is, the complex $3 \times 3$ matrix-valued quantity $\mathcal{J}\colon T^{*}\R^{3} \times \Sp(6,\R) \to {\sf M}_{3}(\C)$ defined by
  \begin{equation}
    \label{eq:semiclassical_angular_momentum-Hagedorn}
   \mathcal{J}(z,Y) \defeq \hat{q}\,P - \hat{p}\,Q
  \end{equation}
  is conserved along the solutions of the semiclassical equations~\eqref{eq:Hagedorn} (or \eqref{eq:Hagedorn2}) of Hagedorn.
\end{example}

\subsection{First Variation Equation and Evolution in Symplectic Group $\Sp(2d,\R)$}
The key to the proof of Theorem~\ref{thm:Noether_for_Hagedorn} is the well-known connection (see, e.g., \citet[Section~2]{Li1986}) between the equation~\eqref{eq:Hagedorn} of Hagedorn and the so-called first variation equation of the classical Hamiltonian system~\eqref{eq:classicalHamSys} on $T^{*}\R^{d}$.
Here we give a brief overview of this result; see also Appendix~\ref{sec:GeomOfFirstVarEq} for those theoretical results that will be used in the proof.

Consider the following linearization {\em along a solution} $z(t) \defeq (q(t), p(t))$ of the classical Hamiltonian system~\eqref{eq:classicalHamSys}:
\begin{equation*}
  \od{}{t}
  \begin{bmatrix}
    \delta q(t) \medskip\\
    \delta p(t)
  \end{bmatrix}
   = 
   \begin{bmatrix}
     D_{1}D_{2}H(z(t)) & D_{2}D_{2}H(z(t)) \medskip\\
     -D_{1}D_{1}H(z(t)) & -D_{2}D_{1}H(z(t))
   \end{bmatrix}
   \begin{bmatrix}
     \delta q(t) \medskip\\
     \delta p(t)
   \end{bmatrix},
\end{equation*}
where $D_{1}$ and $D_{2}$ stand for the derivatives with respect to $q$ and $p$, respectively, of functions of $(q,p)$.
We may rewrite it in a more succinct form
\begin{equation}
  \label{eq:LinearizedSystem}
  \od{}{t}\delta z(t) = \mathbb{J}\,\nabla^{2}H(z(t))\,\delta z(t) = \xi(t)\,\delta z(t),
\end{equation}
where we set
\begin{equation*}
  \delta z(t) \defeq
  \begin{bmatrix}
    \delta q(t) \medskip\\
    \delta p(t)
  \end{bmatrix}
  \in T_{z(t)}(T^{*}\R^{d}) \cong \R^{2d}
\end{equation*}
and $\mathbb{J} \defeq \begin{tbmatrix}
  0 & I_{d} \\
  -I_{d} & 0
\end{tbmatrix}$, and $\xi(t) \in \mathfrak{sp}(2d,\R)$ is defined as
\begin{equation*}
  \xi(t) \defeq \mathbb{J}\,\nabla^{2}H(z(t)) =
  \begin{bmatrix}
    D_{1}D_{2}H(z(t)) & D_{2}D_{2}H(z(t)) \medskip\\
    -D_{1}D_{1}H(z(t)) & -D_{2}D_{1}H(z(t))
  \end{bmatrix}.
\end{equation*}
In particular, with the Hamiltonian of the form
\begin{equation}
  \label{eq:simple_Hamiltonian}
  H = \frac{p^{2}}{2m} + V(q),
\end{equation}
we have the $\xi(t)$ in \eqref{eq:Hagedorn-xi}.

The system consisting of \eqref{eq:classicalHamSys} and \eqref{eq:LinearizedSystem}, i.e.,
\begin{equation}
  \label{eq:FirstVariationEq}
  \dot{z} = \mathbb{J}\nabla H,
  \qquad
  \dot{\delta z} = \xi\,\delta z,
\end{equation}
is called the {\em first variation equation} of \eqref{eq:classicalHamSys}; as shown in Appendix~\ref{sec:GeomOfFirstVarEq} ($\mathcal{P}$ in Appendix~\ref{sec:GeomOfFirstVarEq} is $T^{*}\R^{d}$ here), it is also a Hamiltonian system
\begin{equation*}
  \ins{X_{\tilde{H}}} \Omega_{T(T^{*}\R^{d})} = \d{\tilde{H}}
\end{equation*}
on the tangent bundle
\begin{equation*}
  T(T^{*}\R^{d}) \cong T^{*}\R^{d} \times \R^{2d} \cong \R^{4d} = \{(z, \delta z)\} = \{(q,p,\delta q,\delta p)\}  
\end{equation*}
with the symplectic structure
\begin{equation}
  \label{eq:Omega_TTstarR^d}
  \Omega_{T(T^{*}\R^{d})} =  \d{\delta q} \wedge \d{p} + \d{q} \wedge \d{\delta p}
\end{equation}
and the Hamiltonian $\tilde{H}\colon T(T^{*}\R^{d}) \to \R$ defined by
\begin{equation}
  \label{eq:tildeH}
  \tilde{H}(z,\delta z) \defeq \d{H}(z) \cdot \delta z = \pd{H}{q} \cdot \delta q + \pd{H}{p} \cdot \delta p.
\end{equation}

The Hagedorn equations~\eqref{eq:Hagedorn2} on $T^{*}\R^{d} \times \Sp(2d,\R)$ generalize the first variation equation \eqref{eq:FirstVariationEq} in the sense that \eqref{eq:Hagedorn2} effectively keeps track of solutions of the first variation equation for {\em all} initial conditions $\delta z(0) \in T_{z(0)}(T^{*}\R) \cong \R^{2d}$ at the same time, as opposed to following a single trajectory $\delta z(t)$ for a single particular initial condition $\delta z(0)$:
In fact, for {\em any} $\delta z_{0} \in \R^{2d}$,
\begin{equation}
  \label{eq:deltaz-Y}
  \delta z(t) = Y(t)\,Y(0)^{-1}\,\delta z(0)
\end{equation}
satisfies the linearized equation \eqref{eq:LinearizedSystem}.

Now that the first variation equation is linked with the semiclassical equation of Hagedorn, we exploit the Hamiltonian structure of the first variation equation to formulate a Noether-type theorem for the equations~\eqref{eq:Hagedorn2} of Hagedorn in the presence of symmetry.
The basic idea is to construct conserved quantities from the momentum map of the first variation equation (see Section~\ref{ssec:Symmetry_in_First_Variation_Eq} in Appendix~\ref{sec:GeomOfFirstVarEq}):

\begin{proof}[Proof of Theorem~\ref{thm:Noether_for_Hagedorn}]
  Let $(z(t), Y(t)) = (q(t),p(t),Y(t)) \in T^{*}\R^{d} \times \Sp(2d,\R)$ be a solution of \eqref{eq:Hagedorn2}.
  Setting $\delta z(t) = Y(t)\,\delta z_{0}$ with an {\em arbitrary} $\delta z_{0} \in T_{z(0)}(T^{*}\R) \cong \R^{2d}$, the curve $(z(t),\delta z(t))$ in $T(T^{*}\R^{d})$ is a solution of the first variation equation~\eqref{eq:FirstVariationEq} with the initial condition $(z(0), Y(0)\,\delta z_{0})$. (Note that $Y(0)$ is not necessarily the identity; see \eqref{eq:deltaz-Y}.)
  Now, the assumption on $G$-symmetry implies that this is a special case of the setting discussed in Appendix~\ref{sec:GeomOfFirstVarEq} with $\mathcal{P} = T^{*}\R^{d}$.
  Hence by Proposition~\ref{prop:hatJ}, the momentum map $\tilde{\bf J}\colon T(T^{*}\R^{d}) \to \mathfrak{g}^{*}$ defined by\footnote{We slightly abused the notation here and Appendix~\ref{sec:GeomOfFirstVarEq} and denote by $\delta z$ an element in $T\mathcal{P}$ as well as a tangent vector in $T_{z}P$; what we write $(z,\delta z)$ here is denoted by $\delta z$ in Appendix~\ref{sec:GeomOfFirstVarEq}.}
  \begin{equation*}
    \tilde{{\bf J}}(z,\delta z) \defeq \d{\bf J}(z) \cdot \delta z
  \end{equation*}
  is conserved along the solutions of the first variation equation~\eqref{eq:FirstVariationEq}.
  Therefore,
  \begin{equation*}
    \tilde{{\bf J}}(z(t),Y(t)\,\delta z_{0}) = \d{\bf J}(z(t)) \cdot Y(t)\,\delta z_{0}
  \end{equation*}
  is conserved along the solution $(z(t), Y(t))$.
  However, since $\delta z_{0}$ is chosen arbitrarily, the quantity
  \begin{equation*}
    (D{\bf J}(z) \cdot Y)_{k} = \pd{{\bf J}}{z^{j}}(z(t))\, Y_{jk}(t)
  \end{equation*}
  for $k \in \{1, \dots, d\}$ is conserved along the solution $(z(t), Y(t))$ of the equation~\eqref{eq:Hagedorn2} of Hagedorn.
\end{proof}

\subsection{Rotational Symmetry in the Hagedorn Wave Packet Dynamics}
The above proof, however, does not reveal the group action on the manifold $T^{*}\R^{d} \times \Sp(2d,\R)$.
In this section, we focus on the case with rotational symmetry (see Example~\ref{ex:Noether_for_Hagedorn-SO3}), and find the corresponding $\SO(d)$-action on $T^{*}\R^{d} \times \Sp(2d,\R)$.

Suppose that the potential $V$ is invariant under the $\SO(d)$-action~\eqref{eq:Phi-SO(d)}; then the classical Hamiltonian~\eqref{eq:simple_Hamiltonian} is clearly invariant under the $\SO(d)$-action.
Now the associated tangent $\SO(d)$-action
\begin{equation*}
  T\Phi\colon \SO(d) \times T(T^{*}\R^{d}) \to T(T^{*}\R^{d})
\end{equation*}
is given by
\begin{equation}
  \label{eq:TPhi-SO(d)}
  T\Phi_{R} \colon T(T^{*}\R^{d}) \to T(T^{*}\R^{d});
  \qquad
  (q, p, \delta q, \delta p) \mapsto (R q, R p, R\,\delta q, R\,\delta p).
\end{equation}
It is clearly a symplectic action on $T(T^{*}\R^{d})$ with respect to the symplectic form \eqref{eq:Omega_TTstarR^d}, thus illustrating Lemma~\ref{lemma:Tf}.
Also, by Lemma~\ref{lemma:H-hatH_symmetry}, the Hamiltonian $\tilde{H}$ in \eqref{eq:tildeH} for the first variation equation is $\SO(d)$-invariant as well, i.e., $\tilde{H} \circ T\Phi_{R} = \tilde{H}$.

What is the corresponding $\SO(d)$-action on $\Sp(2d,\R)$?
First recall that we obtained Hagedorn's equations~\eqref{eq:Hagedorn} by defining $Y(t) \in \Sp(2d,\R)$ as
\begin{equation*}
  Y(t) \defeq
  \begin{bmatrix}
    \Re Q(t) & \Im Q(t) \smallskip\\
    \Re P(t) & \Im P(t)
  \end{bmatrix}
\end{equation*}
and setting $\delta z(t) = Y(t)\,\delta z_{0}$ as in \eqref{eq:deltaz-Y} with an {\em arbitrary} $\delta z_{0} \in T_{z(0)}(T^{*}\R) \cong \R^{2d}$.
Since $\delta z$ is transformed under the $\SO(d)$-action~\eqref{eq:TPhi-SO(d)} as
\begin{equation*}
  \delta z =
  \begin{bmatrix}
    \delta q \\
    \delta p
  \end{bmatrix}
  \mapsto
  \begin{bmatrix}
    R\,\delta q \\
    R\,\delta p
  \end{bmatrix}
  = \tilde{R}\,\delta z
  \quad\text{with}\quad
  \tilde{R} \defeq
  \begin{bmatrix}
    R & 0 \\
    0 & R
  \end{bmatrix}
  \in \Sp(2d,\R),
\end{equation*}
we have the actions $\delta z(t) \mapsto \tilde{R}\,\delta z(t)$ and $\delta z_{0} \mapsto \tilde{R}\,\delta z_{0}$, hence the corresponding $\SO(d)$-action $\hat{\gamma}\colon \SO(d) \times \Sp(2d,\R) \to \Sp(2d,\R)$ should satisfy
\begin{equation*}
  \tilde{R}\,\delta z(t) = \hat{\gamma}_{R}(Y(t))\,\tilde{R}\,\delta z_{0},
\end{equation*}
which leads us to the conjugation $\hat{\gamma}_{R}(Y) = \tilde{R} Y \tilde{R}^{T}$, i.e.,
\begin{equation*}
  Y = \begin{bmatrix}
    \Re Q & \Im Q \smallskip\\
    \Re P & \Im P
  \end{bmatrix}
  \mapsto
  \begin{bmatrix}
    R & 0 \\
    0 & R
  \end{bmatrix}
  \begin{bmatrix}
    \Re Q & \Im Q \smallskip\\
    \Re P & \Im P
  \end{bmatrix}
  \begin{bmatrix}
    R^{T} & 0 \\
    0 & R^{T}
  \end{bmatrix}
  =
  \begin{bmatrix}
    R(\Re Q)R^{T} & R(\Im Q)R^{T} \smallskip\\
    R(\Re P)R^{T} & R(\Im P)R^{T}
  \end{bmatrix}.
\end{equation*}

\begin{remark}
  The above $\SO(d)$-action $\hat{\gamma}\colon \SO(d) \times \Sp(2d,\R) \to \Sp(2d,\R)$ defined by
  \begin{equation*}
    \hat{\gamma}_{R}\colon \Sp(2d,\R) \to \Sp(2d,\R);
    \quad
    \begin{bmatrix}
      A & B \\
      C & D
    \end{bmatrix}
    \mapsto
    \begin{bmatrix}
      R A R^{T} & R B R^{T} \\
      R C R^{T} & R D R^{T}
     \end{bmatrix}.
  \end{equation*}
  is compatible with the action on $\Sigma_{d}$ defined in \eqref{eq:SO(d)-action_Sigma_d}, i.e., the diagram
  \begin{equation*}
    \begin{tikzcd}[column sep=8ex, row sep=7ex]
      \Sp(2d,\R) \arrow{r}{} \arrow{d}[swap]{\pi_{\U(d)}} \arrow{r}{\hat{\gamma}_{R}} & \Sp(2d,\R) \arrow{d}{\pi_{\U(d)}}
      \\
      \Sigma_{d} \arrow{r}[swap]{\gamma_{R}} & \Sigma_{d}
    \end{tikzcd}
  \end{equation*}
  commutes for any $R \in \SO(d)$, where $\pi_{\U(d)}\colon \Sp(2d,\R) \to \Sigma_{d}$ is the quotient map defined in \eqref{eq:pi_Ud} of Appendix~\ref{sec:Sigma_d}.
\end{remark}
The cotangent lift~\eqref{eq:Phi-SO(d)} combined with the above action induces an $\SO(d)$-action on $T^{*}\R^{d} \times \Sp(2d,\R)$:
For any $(R, q_{0}) \in \SO(d)$, we define
\begin{equation*}
  \Upsilon_{R}\colon T^{*}\R^{d} \times \Sp(2d,\R) \to T^{*}\R^{d} \times \Sp(2d,\R);
  \quad
  (z, Y) \mapsto (\Phi_{R}(z), \hat{\gamma}_{R}(Y)).
\end{equation*}
When $d = 3$, it is easy to see that the conserved quantity $\mathcal{J}$ in \eqref{eq:semiclassical_angular_momentum-Hagedorn} is equivariant under the natural $\SO(3)$-actions, i.e.,
\begin{equation*}
  \mathcal{J} \circ \Upsilon_{R}(z,Y) = R\,\mathcal{J}(z,Y) R^{T}
\end{equation*}
for any $R \in \SO(3)$.

\section*{Acknowledgments}
This work was inspired by the discussions with Luis Garc\'ia-Naranjo and Joris Vankerschaver at the 2013 SIAM Annual Meeting in San Diego, and was partially supported by the AMS--Simons Travel Grant.

\appendix

\section{The Siegel Upper Half Space $\Sigma_{d}$}
\label{sec:Sigma_d}
\subsection{Geometry of the Siegel Upper Half Space $\Sigma_{d}$}
Recall that the $d \times d$ complex matrix $\mathcal{C} = \mathcal{A} + {\rm i}\mathcal{B}$ in the Gaussian wave packet~\eqref{eq:chi} belongs to the so-called Siegel upper half space $\Sigma_{d}$ defined as (see Eq.~\eqref{eq:Sigma_d})
\begin{equation*}
  \Sigma_{d} \defeq 
  \setdef{ \mathcal{A} + {\rm i}\mathcal{B} \in \mathbb{C}^{d\times d} }{ \mathcal{A}, \mathcal{B} \in \text{Sym}_{d}(\R),\, \mathcal{B} > 0 }.
\end{equation*}
The key to understanding the geometry of the Hagedorn wave packet dynamics in Section~\ref{ssec:Hagedorn_dynamics} is the fact that the Siegel upper half space $\Sigma_{d}$ is a homogeneous space.
Specifically, we can show that (see \citet{Si1943} and also \citet[Section~4.5]{Fo1989} and \citet[Exercise~2.28 on p.~48]{McSa1999})
\begin{equation*}
  \Sigma_{d} \cong \Sp(2d,\R)/\U(d),
\end{equation*}
where $\Sp(2d,\R)$ is the symplectic group of degree $2d$ over real numbers and $\U(d)$ is the unitary group of degree $d$.
In fact, consider the (left) action of $\Sp(2d,\R)$ on $\Sigma_{d}$ defined by
\begin{equation}
  \label{eq:action}
  \Psi\colon \Sp(2d,\R) \times \Sigma_{d} \to \Sigma_{d};
  \quad
  \parentheses{
    \begin{bmatrix}
      A & B \\
      C & D
    \end{bmatrix},
    \mathcal{Z}
  }
  \mapsto
  (C + D\mathcal{Z})(A + B\mathcal{Z})^{-1}.
\end{equation}
This action is transitive: By choosing
\begin{equation*}
  X \defeq
  \begin{bmatrix}
    A & B \\
    C & D
  \end{bmatrix}
  =
  \begin{bmatrix}
    I_{d} & 0 \\
    \mathcal{A} & I_{d}
  \end{bmatrix}
  \begin{bmatrix}
    \mathcal{B}^{-1/2} & 0 \\
    0 & \mathcal{B}^{1/2}
  \end{bmatrix}
  =
  \begin{bmatrix}
    \mathcal{B}^{-1/2} & 0 \\
    \mathcal{A}\mathcal{B}^{-1/2} & \mathcal{B}^{1/2}
  \end{bmatrix},
\end{equation*}
which is easily shown to be symplectic, we have
\begin{equation*}
  \Psi_{X}({\rm i}I_d) = \mathcal{A} + {\rm i}\mathcal{B}.
\end{equation*}
The isotropy group of the element ${\rm i}I_d \in \Sigma_{d}$ is given by
\begin{align*}
  \Sp(2d,\R)_{{\rm i}I_d} &= \setdef{
    \begin{bmatrix}
      U  & V \\
      -V & U
    \end{bmatrix} \in {\sf M}_{2d}(\R)
  }{U^{T}U + V^{T}V = I_{d},\, U^{T}V = V^{T}U}
  \\
  &= \Sp(2d,\R) \cap \mathsf{O}(2d),
\end{align*}
where $\mathsf{O}(2d)$ is the orthogonal group of degree $2d$; however $\Sp(2d,\R) \cap \mathsf{O}(2d)$ is identified with $\U(d)$ as follows:
\begin{equation*}
  \Sp(2d,\R) \cap \mathsf{O}(2d) \to \U(d);
  \quad
  \begin{bmatrix}
    U  & V \\
    -V & U
  \end{bmatrix}
  \mapsto U + {\rm i}\,V.
\end{equation*}
Hence $\Sp(2d,\R)_{{\rm i}I_d} \cong \U(d)$ and thus $\Sigma_{d} \cong \Sp(2d,\R)/\U(d)$.
Indeed, we may identify $\Sp(2d,\R)/\U(d)$ with $\Sigma_{d}$ by the following map:
\begin{equation*}
  \Sp(2d,\R)/\U(d) \to \Sigma_{d};
  \quad
  [Y]_{\U(d)} \mapsto \Psi_{Y}({\rm i}I_d),
\end{equation*}
where $[\,\cdot\,]_{\U(d)}$ stands for a left coset of $\U(d)$ in $\Sp(2d,\R)$; then this gives rise to the explicit construction of the quotient map
\begin{equation}
  \label{eq:pi_Ud}
  \pi_{\U(d)}\colon \Sp(2d,\R) \to \Sp(2d,\R)/\U(d) \cong \Sigma_{d};
  \quad
  Y \mapsto \Psi_{Y}({\rm i}I_d),
\end{equation}
or more specifically,
\begin{equation*}
  \pi_{\U(d)}\parentheses{
    \begin{bmatrix}
      A & B \\
      C & D
    \end{bmatrix}
  }
  = (C + {\rm i}D)(A + {\rm i}B)^{-1}.
\end{equation*}
Therefore, we have the following diagram, which simply shows that the action $\Psi$ is indeed a left action:
Note that the map $L_{X}\colon \Sp(2d,\R) \to \Sp(2d,\R)$ is the standard matrix multiplication from the left by $X$.
\begin{equation}
  \label{cd:Sp-Sigma_d}
  \begin{tikzcd}[column sep=7ex, row sep=7ex]
    \Sp(2d,\R) \arrow{r}{} \arrow{d}[swap]{\pi_{\U(d)}} \arrow{r}{L_{X}} & \Sp(2d,\R) \arrow{d}{\pi_{\U(d)}}
    \\
    \Sigma_{d} \arrow{r}[swap]{\Psi_{X}} & \Sigma_{d}
  \end{tikzcd}
  \qquad
  \begin{tikzcd}[column sep=7ex, row sep=7ex]
    Y \arrow[mapsto]{r} \arrow[mapsto]{d} & X\,Y \arrow[mapsto]{d}
    \\
    \Psi_{Y}({\rm i}I_d) \arrow[mapsto]{r} & \Psi_{X} \circ \Psi_{Y}({\rm i}I_d)
  \end{tikzcd}
\end{equation}
As shown by \citet{Si1943}, the map $\Psi_{X}\colon \Sigma_{d} \to \Sigma_{d}$ is an isometry of the Hermitian metric~\eqref{eq:metric-Sigma_d} for any $X \in \Sp(2d,\R)$ and therefore is symplectic with respect to the symplectic form~\eqref{eq:symplectic_form-Sigma_d}.
This suggests that the $\Sigma_{d}$-component of the symplectic dynamics defined by the reduced semiclassical equations~\eqref{eq:Heller-reduced} may be lifted to the symplectic group $\Sp(2d,\R)$; see Proposition~\ref{prop:Hagedorn-Heller} below.

\subsection{Connection between Symplectic and Hagedorn Semiclassical Dynamics}
The geometry of the Siegel upper half space described above gives rise to a connection between the $\Sigma_{d}$-component of the semiclassical equations~\eqref{eq:Heller-asymptotic-reduced} and the $\Sp(2d,\R)$-component of the equations~\eqref{eq:Hagedorn-Sp} of Hagedorn:
\begin{proposition}
  \label{prop:Hagedorn-Heller}
  The $\Sigma_{d}$-component of the semiclassical equations~\eqref{eq:Heller-asymptotic-reduced}, i.e.,
  \begin{equation}
    \dot{\mathcal{C}} = -\frac{1}{m}\mathcal{C}^{2} - \nabla^{2}V(q)
    \tag{\ref{eq:Riccati-C}}
  \end{equation}
  is the projection by the quotient map $\pi_{\U(d)}\colon \Sp(2d,\R) \to \Sigma_{d}$ to the Siegel upper half space $\Sigma_{d}$ of a curve
  \begin{equation*}
    Y(t) =
    \begin{bmatrix}
      \Re Q(t) & \Im Q(t) \smallskip\\
      \Re P(t) & \Im P(t)
    \end{bmatrix}
  \end{equation*}
  in the symplectic group $\Sp(2d,\R)$ defined by
  \begin{equation}
    \dot{Y}(t) = \xi(t)\,Y(t)
    \tag{\ref{eq:Hagedorn-Sp}}
  \end{equation}
  with $\xi(t) \in \mathfrak{sp}(2d,\R)$ being
  \begin{equation}
    \xi(t) \defeq 
    \begin{bmatrix}
      0 & I_{d}/m \smallskip\\
      -\nabla^{2}V(q(t)) & 0
    \end{bmatrix},
    \tag{\ref{eq:Hagedorn-xi}}
  \end{equation}
  or equivalently,
  \begin{equation*}
    \dot{Q} = \frac{P}{m},
    \qquad
    \dot{P} =  -\nabla^{2}V(q)\,Q.
  \end{equation*}
  Furthermore, the lift \eqref{eq:Hagedorn-Sp} is unique in the following sense: For the vector field on $\Sp(2d,\R)$ defined by \eqref{eq:Hagedorn-Sp} to project via $\pi_{\U(d)}$ to \eqref{eq:Riccati-C} on $\Sigma_{d}$ for any $\mathcal{C} \in \Sigma_{d}$, $\xi(t)$ has to take the form~\eqref{eq:Hagedorn-xi}.
\end{proposition}

\begin{proof}
  Let $\mathcal{C}(t)$ be a curve in $\Sigma_{d}$ defined by \eqref{eq:Riccati-C}.
  As shown in the previous subsection, the action of the symplectic group $\Sp(2d,\R)$ on $\Sigma_{d}$ defined as
  \begin{equation}
    \label{eq:Psi}
    \Psi\colon \Sp(2d,\R) \times \Sigma_{d} \to \Sigma_{d};
    \quad
    \parentheses{
      \begin{bmatrix}
        A & B \\
        C & D
      \end{bmatrix},
      \mathcal{Z}
    }
    \mapsto
    (C + D\mathcal{Z})(A + B\mathcal{Z})^{-1}
  \end{equation}
  is transitive.
  Therefore, there exists a corresponding curve $X(t)$ in $\Sp(2d,\R)$ such that $\Psi_{X(t)}(\mathcal{C}(0)) = \mathcal{C}(t)$ and $X(0) = I_{d}$.
  Now, let $Y_{0} \in \Sp(2d,\R)$ be an element such that $\pi_{\U(d)}(Y_{0}) = \mathcal{C}(0)$, and define the curve $Y(t) \defeq X(t) Y_{0}$.
  Then clearly we have $\pi_{\U(d)} \circ Y(t) = \mathcal{C}(t)$, i.e., the following diagram commutes as in \eqref{cd:Sp-Sigma_d}.
  \begin{equation*}
    \begin{tikzcd}[column sep=7ex, row sep=7ex]
      \Sp(2d,\R) \arrow{r}{} \arrow{d}[swap]{\pi_{\U(d)}} \arrow{r}{L_{X(t)}} & \Sp(2d,\R) \arrow{d}{\pi_{\U(d)}}
      \\
      \Sigma_{d} \arrow{r}[swap]{\Psi_{X(t)}} & \Sigma_{d}
    \end{tikzcd}
    \qquad
    \begin{tikzcd}[column sep=7ex, row sep=7ex]
      Y_{0} \arrow[mapsto]{r} \arrow[mapsto]{d} & X(t) Y_{0} \arrow[mapsto]{d}
      \\
      \mathcal{C}(0) \arrow[mapsto]{r} & \mathcal{C}(t)
    \end{tikzcd}
  \end{equation*}
  Let us then write
  \begin{equation*}
    \xi(t) \defeq \dot{Y}(t)\,Y(t)^{-1} = \dot{X}(t)\,X(t)^{-1},
  \end{equation*}
  which is in the Lie algebra $\mathfrak{sp}(2d,\R)$; thus it takes the form
  \begin{equation*}
    \xi(t) = 
    \begin{bmatrix}
      \xi_{11}(t) & \xi_{12}(t) \smallskip\\
      \xi_{21}(t) & -\xi_{11}(t)^{T}
    \end{bmatrix},
  \end{equation*}
  where $\xi_{ij}$ with $i,j \in \{1, 2\}$ are all $d \times d$ real matrices and $\xi_{12}$ and $\xi_{21}$ are both symmetric; then it is easy to see that $\dot{Y} = \xi Y$ gives
  \begin{equation*}
    \dot{Q} = \xi_{11} Q + \xi_{12} P,
    \qquad
    \dot{P} = \xi_{21} Q - \xi_{11}^{T} P.
  \end{equation*}
  Therefore, from \eqref{eq:pi_Ud-2} and the above expressions, we have
  \begin{align*}
    T_{Y}\pi_{\U(d)}(\dot{Y}) &= \dot{P} Q^{-1} - P Q^{-1} \dot{Q} Q^{-1}
    \\
    &=  \xi_{21} - \xi_{11}^{T}\,\mathcal{C} - \mathcal{C}\,\xi_{11} - \mathcal{C}\,\xi_{12}\,\mathcal{C}.
  \end{align*}
  where we also used the relation $\pi_{\U(d)} \circ Y(t) = \mathcal{C}(t)$, i.e., $P Q^{-1} = \mathcal{C}$.
  However, taking the time derivative of $\pi_{\U(d)} \circ Y(t) = \mathcal{C}(t)$, we have $T\pi_{\U(d)}(\dot{Y}) = \dot{\mathcal{C}}$, which implies, using the above expression and \eqref{eq:Riccati-C},
  \begin{equation}
    \label{eq:mathcalC_equality}
    \xi_{21} - \xi_{11}^{T}\,\mathcal{C} - \mathcal{C}\,\xi_{11} - \mathcal{C}\,\xi_{12}\,\mathcal{C}
    =
    -\frac{1}{m}\mathcal{C}^{2} - \nabla^{2}V(q).
  \end{equation}
  Now, let us find the entries for $\xi(t)$ such that the above equality holds for any $\mathcal{C} \in \Sigma_{d}$.
  Setting $\mathcal{C} = {\rm i}I_{d}$ in the above equality~\eqref{eq:mathcalC_equality} gives
  \begin{equation*}
    \xi_{12} + \xi_{21} = \frac{1}{m}\,I_{d} - \nabla^{2}V(q),
    \qquad
    \xi_{11} + \xi_{11}^{T} = 0,
  \end{equation*}
  whereas setting $\mathcal{C} = 2{\rm i}I_{d}$ in \eqref{eq:mathcalC_equality} gives
  \begin{equation*}
    4\xi_{12} + \xi_{21} = \frac{4}{m}\,I_{d} - \nabla^{2}V(q),
    \qquad
    \xi_{11} + \xi_{11}^{T} = 0,
  \end{equation*}
  Therefore, we have
  \begin{equation*}
    \xi_{12} = \frac{1}{m}\,I_{d},
    \qquad
    \xi_{21} = -\nabla^{2}V,
    \qquad
    \xi_{11}^{T} = -\xi_{11}.
  \end{equation*}
  So \eqref{eq:mathcalC_equality} now reduces to 
  \begin{equation*}
    \xi_{11}\,\mathcal{C} - \mathcal{C}\,\xi_{11} = 0.
  \end{equation*}
  Writing $\mathcal{C} = \mathcal{A} + {\rm i}\mathcal{B}$ and taking the real part:
  \begin{equation*}
    \xi_{11}\,\mathcal{A} - \mathcal{A}\,\xi_{11} = 0,
  \end{equation*}
  where $\mathcal{A}$ is an arbitrary $d \times d$ symmetric matrix.
  Setting $\mathcal{A} = e_{j} e_{j}^{T}$, with $e_{j} \in \R^{d}$ being the unit vector whose $j$-th entry is 1, shows that the $j$-th column of $\xi_{11}$ is 0.
  Since $j \in \{1, \dots, n\}$ is taken arbitrary, we have $\xi_{11} = 0$.
  Hence we have \eqref{eq:Hagedorn-xi}.
  It is clear that \eqref{eq:mathcalC_equality} holds for any $\mathcal{C} \in \Sigma_{d}$ with $\xi$ taking the form \eqref{eq:Hagedorn-xi}.
\end{proof}

\begin{remark}
  One may recognize \eqref{eq:Riccati-C} as an example of the matrix Riccati equation.
  In fact, there is a similar geometric structure behind the matrix Riccati equation: As shown in \citet{HeMa1977} and \citet{DoMa1990}, one considers the action of a general linear group on a Grassmannian using the linear fractional transformation of the form \eqref{eq:Psi} (see also \eqref{eq:action}), and then performs virtually the same calculations as above to derive the matrix Riccati equation.

  The above lift may also be regarded as an example of the Hirota bilinearization of the matrix Riccati equation; see, e.g., \citet{Hi1979, Hi2000, Hi2004}.
\end{remark}

\section{Geometry of the First Variation Equation}
\label{sec:GeomOfFirstVarEq}
This section gives a brief summary of the geometry of the first variation equation.
Our main references are \citet{Tu1976a}, \citet{SnTu1972a}, \citet{AbMa1978}, and \citet{MaRaRa1991}.

\subsection{Symplectic Structure for the First Variation Equation}
Let $\mathcal{P}$ be a symplectic manifold with symplectic form $\Omega_{\mathcal{P}}$ and $H\colon \mathcal{P} \to \R$ be a Hamiltonian, and define the Hamiltonian system
\begin{equation}
  \label{eq:HamiltonianSystem_on_P}
  \ins{X_{H}}\Omega_{\mathcal{P}} = \d{H}
\end{equation}
on $\mathcal{P}$, where $X_{H}$ is the corresponding Hamiltonian vector field on $\mathcal{P}$.
Let $\phi_{t}\colon \mathcal{P} \to \mathcal{P}$ be the flow defined by the vector field $X_{H}$, i.e., for any $z \in \mathcal{P}$,
\begin{equation*}
  \left.\od{}{t}\phi_{t}(z)\right|_{t=0} = X_{H}(z).
\end{equation*}
Let $\tau_{\mathcal{P}}\colon T\mathcal{P} \to \mathcal{P}$ be the tangent bundle of $\mathcal{P}$; then $T\mathcal{P}$ is an example of a {\em special symplectic manifold} (see \cite{Tu1976a} and also \cite{SnTu1972a} and \cite[Exercise~3.3I on p.~200]{AbMa1978}) and is an exact symplectic manifold with the symplectic form $\Omega_{T\mathcal{P}} \defeq -\d\Theta_{T\mathcal{P}}$ with the canonical one-form $\Theta_{T\mathcal{P}}$ on $T\mathcal{P}$ defined as follows: For any $\delta z \in T\mathcal{P}$ and $v_{\delta z} \in T_{\delta z}(T\mathcal{P})$,
\begin{equation}
  \label{eq:Theta_TP}
  \Theta_{T\mathcal{P}}(\delta z) \cdot v_{\delta z} = \Omega_{\mathcal{P}}(T\tau_{\mathcal{P}}(v_{\delta z}), \delta z).
\end{equation}
This canonical one-form has the following nice property:

\begin{lemma}
  \label{lemma:Tf}
  If $f\colon \mathcal{P} \to \mathcal{P}$ is symplectic, then its tangent map $Tf\colon T\mathcal{P} \to T\mathcal{P}$ preserves the canonical one-form $\Theta_{T\mathcal{P}}$, i.e., $(Tf)^{*} \Theta_{T\mathcal{P}} = \Theta_{T\mathcal{P}}$, and hence is symplectic with respect to $\Omega_{T\mathcal{P}} \defeq -\d\Theta_{T\mathcal{P}}$.
\end{lemma}

\begin{proof}
  For any $\delta z \in T\mathcal{P}$ and $v_{\delta z} \in T_{\delta z}(T\mathcal{P})$, 
  \begin{align*}
    (Tf)^{*}\Theta_{T\mathcal{P}}(\delta z) \cdot v_{\delta z} &= \Theta_{T\mathcal{P}}(Tf(\delta z)) \cdot TTf(v_{\delta z})
    \\
    &= \Omega_{\mathcal{P}}(T\tau_{\mathcal{P}} \circ TTf(v_{\delta z}), Tf(\delta z))
    \\
    &= \Omega_{\mathcal{P}}(T(\tau_{\mathcal{P}} \circ Tf)(v_{\delta z}), Tf(\delta z))
    \\
    &= \Omega_{\mathcal{P}}(Tf \circ T\tau_{\mathcal{P}}(v_{\delta z}), Tf(\delta z))
    \\
    &= f^{*}\Omega_{\mathcal{P}}(T\tau_{\mathcal{P}}(v_{\delta z}), \delta z)
    \\
    &= \Omega_{\mathcal{P}}(T\tau_{\mathcal{P}}(v_{\delta z}), \delta z)
    \\
    &= \Theta_{T\mathcal{P}}(\delta z) \cdot v_{\delta z},
  \end{align*}
  where we used the identity $\tau_{\mathcal{P}} \circ Tf = f \circ \tau_{\mathcal{P}}$ and symplecticity of $f$ with respect to $\Omega_{\mathcal{P}}$.
\end{proof}

Now, consider the tangent map of the flow $T\phi_{t}\colon T\mathcal{P} \to T\mathcal{P}$ and let $\tilde{X}_{H}$ be the vector field on $T\mathcal{P}$ defined by the flow $T\phi_{t}$, i.e., for any $\delta z \in T\mathcal{P}$,
\begin{equation*}
  \tilde{X}_{H}(\delta z) \defeq \left.\od{}{t}T\phi_{t}(\delta z)\right|_{t=0}.
\end{equation*}

\begin{proposition}[{See, e.g., \cite[Exercise~3.8E on p.~252]{AbMa1978} and references therein}]
  The vector field $\tilde{X}_{H}$ on $T\mathcal{P}$ is the Hamiltonian vector field on $T\mathcal{P}$ with respect to the symplectic form $\Omega_{T\mathcal{P}}$ and the Hamiltonian $\tilde{H}\colon T\mathcal{P} \to \R$ defined by
  \begin{equation}
    \label{eq:hatH}
    \tilde{H}(\delta z) \defeq \d{H}(z) \cdot \delta z
  \end{equation}
  for any $\delta z \in T_{z}\mathcal{P}$, where $z = \tau_{\mathcal{P}}(\delta z)$; that is, we have $\tilde{X}_{H} = X_{\tilde{H}}$, where $X_{\tilde{H}}$ is the Hamiltonian vector field on $T\mathcal{P}$ for the above Hamiltonian $\tilde{H}$, i.e.,
  \begin{equation}
    \label{eq:HamiltonianSystem-X_hatH}
    \ins{X_{\tilde{H}}} \Omega_{T\mathcal{P}} = \d{\tilde{H}}.
  \end{equation}
\end{proposition}

\begin{proof}
  Taking the time derivative at $t = 0$ of the identity
  \begin{equation*}
    \tau_{\mathcal{P}} \circ T\phi_{t}(\delta z) = \phi_{t}(z),
  \end{equation*}
  we obtain
  \begin{equation*}
    T\tau_{\mathcal{P}} \circ \tilde{X}_{H}(\delta z) = X_{H}(z).
  \end{equation*}
  Therefore, from the definition~\eqref{eq:Theta_TP} of the canonical one-form $\Theta_{T\mathcal{P}}$
  \begin{align*}
    \ins{\tilde{X}_{H}} \Theta_{T\mathcal{P}}(\delta z)
    &= \Omega_{\mathcal{P}}\parentheses{ T\tau_{\mathcal{P}} \circ \tilde{X}_{H}(\delta z), \delta z }
    \\
    &= \Omega_{\mathcal{P}}\parentheses{ X_{H}(\delta z), \delta z }
    \\
    &= \d{H}(z) \cdot \delta z.
    \\
    &= \tilde{H}(\delta z),
  \end{align*}
  where we used the definition~\eqref{eq:hatH} of $\tilde{H}$.
  Taking the exterior differential of the above, we have $\d\,\ins{\tilde{X}_{H}} \Theta_{T\mathcal{P}} = \d\tilde{H}$.
  However, since $\phi_{t}$ is symplectic, we have $(T\phi_{t})^{*} \Theta_{T\mathcal{P}} = \Theta_{T\mathcal{P}}$ by Lemma~\ref{lemma:Tf} and so $\pounds_{\!\tilde{X}_{H}} \Theta_{T\mathcal{P}} = 0$; then Cartan's formula gives $\d\,\ins{\!\tilde{X}_{H}}\Theta_{T\mathcal{P}} = -\ins{\!\tilde{X}_{H}}\d\Theta_{T\mathcal{P}} = \ins{\!\tilde{X}_{H}}\Omega_{T\mathcal{P}}$.
\end{proof}

\subsection{Local Expressions and the First Variation Equation}
Canonical coordinates $z = (q,p)$ for $\mathcal{P}$ induces the coordinates $\delta z = (q, p, \delta q, \delta p)$ on $T\mathcal{P}$ and then we have
\begin{equation*}
  \Theta_{T\mathcal{P}}(\delta z) = \delta p \cdot \d{q} - \delta q \cdot \d{p}
\end{equation*}
and
\begin{equation*}
  \Omega_{T\mathcal{P}}(\delta z) =  \d{\delta q} \wedge \d{p} + \d{q} \wedge \d{\delta p}.
\end{equation*}
Also, the Hamiltonian $\tilde{H}$ can be written as follows:
\begin{equation*}
  \tilde{H}(\delta z) = \pd{H}{q} \cdot \delta q + \pd{H}{p} \cdot \delta p.
\end{equation*}
Then the Hamiltonian system $\ins{\tilde{X}_{H}} \Omega_{T\mathcal{P}} = \d{\tilde{H}}$ gives the first variation equation, i.e., the classical Hamiltonian system~\eqref{eq:classicalHamSys} with the linearized system~\eqref{eq:LinearizedSystem} along its solution.

\subsection{Symmetry and Conservation Laws in the First Variation Equation}
\label{ssec:Symmetry_in_First_Variation_Eq}
Let $G$ be a Lie group and consider its symplectic action $\Phi\colon G \times \mathcal{P} \to \mathcal{P}$ on the symplectic manifold $\mathcal{P}$, and suppose that the Hamiltonian system~\eqref{eq:HamiltonianSystem_on_P} (or locally \eqref{eq:classicalHamSys}) on $\mathcal{P}$ has a $G$-symmetry.
Then one can easily show that the first variation equation~\eqref{eq:HamiltonianSystem-X_hatH} also has a symmetry under the action $T\Phi\colon G \times T\mathcal{P} \to T\mathcal{P}$ induced by the tangent map $T\Phi_{g}$ of $\Phi_{g}$:

\begin{lemma}
  \label{lemma:H-hatH_symmetry}
  If the Hamiltonian $H\colon \mathcal{P} \to \R$ is invariant under the $G$-action, i.e., $H \circ \Phi_{g} = H$ for any $g \in G$, then so is the Hamiltonian $\tilde{H}\colon T\mathcal{P} \to \R$, i.e., $\tilde{H} \circ T\Phi_{g} = \tilde{H}$ for any $g \in G$.
\end{lemma}

\begin{proof}
  Follows easily from the following simple calculations: For any $\delta z \in T\mathcal{P}$ with $z = \tau_{\mathcal{P}}(\delta z)$ and any $g \in G$,
  \begin{align*}
    \tilde{H} \circ T\Phi_{g}(\delta z) &= \d{H}(\Phi_{g}(z)) \cdot T\Phi_{g}(\delta z)
    \\
    &= (\Phi_{g}^{*}\d{H})(z) \cdot \delta z
    \\
    &= \d(\Phi_{g}^{*}{H})(z) \cdot \delta z
    \\
    &= \d{H}(z) \cdot \delta z
    \\
    &= \tilde{H}(\delta z). \qedhere
  \end{align*}
\end{proof}

Now we are ready to state Noether's theorem for the first variation equation:
\begin{proposition}
  \label{prop:hatJ}
  Suppose that the Hamiltonian $H\colon \mathcal{P} \to \R$ is invariant under the $G$-action, and let ${\bf J}\colon \mathcal{P} \to \mathfrak{g}^{*}$ be the corresponding momentum map; then the momentum map $\tilde{\bf J}\colon T\mathcal{P} \to \mathfrak{g}^{*}$ corresponding to the induced action $T\Phi$ on $T\mathcal{P}$ is given by
  \begin{equation*}
    \tilde{{\bf J}}(\delta z) \defeq \d{\bf J}(z) \cdot \delta z,
  \end{equation*}
  and $\tilde{{\bf J}}$ is conserved the along the flow of the first variation equation~\eqref{eq:HamiltonianSystem-X_hatH}.
\end{proposition}

\begin{proof}
  Let $\xi$ be an arbitrary element in the Lie algebra $\mathfrak{g}$ of $G$, and $\xi_{\mathcal{P}}$ and $\xi_{T\mathcal{P}}$ be its infinitesimal generators on $\mathcal{P}$ and $T\mathcal{P}$, respectively, i.e.,
  \begin{equation*}
    \xi_{\mathcal{P}}(z) \defeq \left. \od{}{t}\Phi_{\exp(t\xi)}(z) \right|_{t=0},
    \qquad
    \xi_{T\mathcal{P}}(\delta z) \defeq \left. \od{}{t}T\Phi_{\exp(t\xi)}(\delta z) \right|_{t=0}.
  \end{equation*}
  By taking the time derivative at $t = 0$ of the identity
  \begin{equation*}
    \tau_{\mathcal{P}} \circ T\Phi_{\exp(t\xi)}(\delta z) = \Phi_{\exp(t\xi)} \circ \tau_{\mathcal{P}}(\delta z),
  \end{equation*}
  we have
  \begin{equation}
    \label{eq:xi_TP-xi_P}
    T\tau_{\mathcal{P}} \circ \xi_{T\mathcal{P}} = \xi_{\mathcal{P}} \circ \tau_{\mathcal{P}}.
  \end{equation}
  Also recall that the momentum map ${\bf J}\colon \mathcal{P} \to \mathfrak{g}^{*}$ satisfies
  \begin{equation}
    \label{eq:J-Hamiltonian}
    \ins{\xi_{\mathcal{P}}} \Omega_{\mathcal{P}}(z) = \d\ip{ {\bf J}(z) }{ \xi }.
  \end{equation}
  Now, since $\Phi$ is a symplectic action, $T\Phi$ preserves the canonical one-form $\Theta_{T\mathcal{P}}$ by Lemma~\ref{lemma:Tf}; hence we can calculate the momentum map $\tilde{\bf J}\colon T\mathcal{P} \to \mathfrak{g}^{*}$ as follows~(see, e.g., \cite[Theorem~4.2.10 on p.~282]{AbMa1978}):
  \begin{align*}
    \ip{ \tilde{\bf J}(\delta z) }{ \xi } &= \ins{\xi_{T\mathcal{P}}}\Theta_{T\mathcal{P}}(\delta z).
    \\
    &= \Omega_{\mathcal{P}}(T\tau_{\mathcal{P}} \circ \xi_{T\mathcal{P}}(\delta z), \delta z)
    \\
    &= \Omega_{\mathcal{P}}(\xi_{\mathcal{P}}(z), \delta z)
    \\
    &= \d\ip{ {\bf J}(z) }{ \xi } \cdot \delta z
    \\
    &= \ip{ \d{\bf J}(z) \cdot \delta z }{ \xi },
  \end{align*}
  where we also used \eqref{eq:xi_TP-xi_P} and \eqref{eq:J-Hamiltonian}.
\end{proof}

\bibliography{SymInSemiClDyn}
\bibliographystyle{plainnat}

\end{document}